\documentclass[11pt]{article}

\usepackage[top=1in,bottom=1in,left=1in,right=1in]{geometry}
\usepackage{natbib}
\usepackage[unicode=true,bookmarks=true,bookmarksnumbered=false,bookmarksopen=false,breaklinks=false,pdfborder={0 0 1},backref=false,colorlinks=true]{hyperref}
\hypersetup{citecolor=blue}
\usepackage{url}

\usepackage{amsmath}
\usepackage{mathrsfs}
\usepackage{amssymb}
\usepackage{amsthm}
\usepackage[normalem]{ulem}

\newtheorem{theorem}{Theorem}
\newtheorem{lemma}{Lemma}

\newtheorem{proposition}{Proposition}

\newtheorem{definition}{Definition}

\usepackage{float}
\usepackage{rotfloat}

\floatstyle{ruled}
\newfloat{algorithm}{tpb}{loa}
\providecommand{\algorithmname}{Algorithm}
\floatname{algorithm}{\protect\algorithmname}

\usepackage{algpseudocode,algorithm}
\algnewcommand\algorithmicinput{\textbf{Input}:}
\algnewcommand\algorithmicoutput{\textbf{Output}:}
\algnewcommand\INPUT{\item[\algorithmicinput]}
\algnewcommand\OUTPUT{\item[\algorithmicoutput]}

\usepackage{graphicx}
\usepackage[caption=false,labelformat=simple]{subfig}

\usepackage{array}
\newcolumntype{L}[1]{>{\raggedright\let\newline\\\arraybackslash\hspace{0pt}}m{#1}}
\newcolumntype{C}[1]{>{\centering\let\newline\\\arraybackslash\hspace{0pt}}m{#1}}
\newcolumntype{R}[1]{>{\raggedleft\let\newline\\\arraybackslash\hspace{0pt}}m{#1}}

\usepackage{rotating}
\usepackage{multirow}

\usepackage{tikz}
\usetikzlibrary{shapes,arrows,backgrounds,calc,positioning,fit,petri,plotmarks}
\usetikzlibrary{arrows}

\usepackage{enumerate}
\usepackage{paralist}

\newcommand*{\affaddr}[1]{#1} 
\newcommand*{\affmark}[1][*]{\textsuperscript{#1}}

\global\long\def\bx{\mathbf{x}}

\global\long\def\bbeta{\boldsymbol{\beta}}

\title{Density-on-Density Regression}

\author{%
    Yi Zhao\affmark[1], Abhirup Datta\affmark[2], Bohao Tang\affmark[2], Vadim Zipunnikov\affmark[2],  Brian S. Caffo\affmark[2], \\ 
    and for the Alzheimer's Disease Neuroimaging Initiative\footnote{Data used in preparation of this article were obtained from the Alzheimer's Disease Neuroimaging Initiative (ADNI) database (\url{adni.loni.usc.edu}). As such, the investigators within the ADNI contributed to the design and implementation of ADNI and/or provided data but did not participate in analysis or writing of this report. A complete list of ADNI investigators can be found at: \url{http://adni.loni.usc.edu/wp-content/uploads/how_to_apply/ADNI_Acknowledgement_List.pdf}} \\
    \affaddr{\affmark[1]Department of Biostatistics and Health Data Science, Indiana University School of Medicine} \\
    \affaddr{\affmark[2]Department of Biostatistics, Johns Hopkins Bloomberg School of Public Health} \\
}

\date{}

\providecommand{\keywords}[1]
{
  {\small	
  \textbf{Keywords:} #1 }
}

\begin{document}

\maketitle

\thispagestyle{empty}

\begin{abstract}

In this study, a density-on-density regression model is introduced, where the association between densities is elucidated via a warping function. The proposed model has the advantage of a being straightforward demonstration of how one density transforms into another. Using the Riemannian representation of density functions, which is the square-root function (or half density), the model is defined in the correspondingly constructed Riemannian manifold. To estimate the warping function, it is proposed to minimize the average Hellinger distance, which is equivalent to minimizing the average Fisher-Rao distance between densities. An optimization algorithm is introduced by estimating the smooth monotone transformation of the warping function. Asymptotic properties of the proposed estimator are discussed. Simulation studies demonstrate the superior performance of the proposed approach over competing approaches in predicting outcome density functions. Applying to a proteomic-imaging study from the Alzheimer's Disease Neuroimaging Initiative, the proposed approach illustrates the connection between the distribution of protein abundance in the cerebrospinal fluid and the distribution of brain regional volume. Discrepancies among cognitive normal subjects, patients with mild cognitive impairment, and Alzheimer's disease (AD) are identified and the findings are in line with existing knowledge about AD.
\end{abstract}




\clearpage
\setcounter{page}{1}

\section{Introduction}
\label{sec:intro}

In this manuscript, a density-on-density regression model is introduced that handles subject-specific density outcomes and density predictors. 
Regression analysis is a fundamental tool in statistical modeling to study the association between two objects, the outcome variable and the predictor variable. The most common form is the linear regression model assuming the two objects are in Euclidean space. Particularly, denote $y\in\mathbb{R}$ as the response variable and $\bx\in\mathbb{R}^{p}$ as the explanatory variable, a liner regression model has the following form
\begin{equation}\label{eq:linear_reg}
  \mathbb{E}\left(y\mid\bx\right)=\bx^\top\bbeta,
\end{equation}
where $\bbeta\in\mathbb{R}^{p}$ is the model coefficient. Extending to a Hilbert space, functional regression models were introduced~\citep{ramsay2005functional}. Properties and various extensions were well studied~\cite[see review articles by][and among others]{morris2015functional,wang2016functional}. With the advancement of modern technologies, data in more complex forms, such as data objects in Riemannian spaces, are collected across various scientific domains, including medical imaging, computational biology, computer vision, digital health technologies, and many others. Extending regression models to deal with these novel data types attract increasing attention~\citep{davis2010population,cornea2017regression,dai2021modeling}. In this study, we focus on a regression scenario where both outcome and predictor are probability densities on the real line satisfying certain regularity conditions.

The study of random densities is challenging considering the fact that a probability density is both functional and non-Euclidean, where the non-linear constraints require the function to be positive and with a unit integral. An extension of Model~\eqref{eq:linear_reg} to the scenario where both $y$ and $x$ are density functions is not straightforward, due to the lack of a linear structure, where most basic statistical properties and theories rely on~\citep{petersen2022modeling}. Equivalent to studying the density functions, most recent research focuses on the study of probability distributions, where non-linear constraints are embedded as well.
One direction of addressing such an obstacle is to utilize an appropriate transformation on the distributions mapping the problem to Hilbert spaces and handling it using existing techniques~\citep{kneip2001inference,petersen2016functional}. Example transformations include log quantile density, log hazard function, and many others. 
\citet{yang2020quantile} proposed quantile function outcome regressions with scalar predictors that were conceptually similar to functional principal component regression. To guarantee non-decreasing monotonicity of quantile function outcome, \citet{yang2020random} imposed monotonicity on individual functional regression coefficients by employing I-splines.
An important limitation of transformation approaches is the lack of isometry which leads to deformations in natural geometry and changes in distances between data objects.
Another direction is to perform distribution regression in its native space equipped with a proper distance metric. Examples of distance metrics include the Fisher-Rao metric~\citep{srivastava2007riemannian} and the Wasserstein metric~\citep{panaretos2020invitation}, each tied to a manifold structure on distributions.
Considering Riemannian structures, the Fisher-Rao metric is a natural choice, as it is the only metric invariant to re-parameterizations on the functions that form a manifold~\citep{cencov1982statistical}. For density functions, the corresponding Riemannian representation is the square root function and the Fisher-Rao distance is the spherical geodesic distance between square root densities~\citep{srivastava2007riemannian}.
The Wasserstein distance is an optimal transport metric for distributions. Based on the fact that the tangent space of probability distributions is a subspace of the infinite-dimensional Hilbert space, \citet{chen2021wassersteinreg} introduced a distribution-on-distribution regression model with an isometric mapping. Utilizing the geometric properties, the asymptotic properties were studied under both the Wasserstein metric and parallel transport. A Wasserstein autoregressive model for density time series was also introduced independently by \citet{zhang2022wasserstein} around the same time. 
These two approaches offer a well-developed toolbox and pave the theoretical foundation for distribution-on-distribution regression in Wasserstein space. However, both suffer from a straightforward interpretation of the regression model lacking connections between distributions in a point-by-point sense.
Following a shape-constraint approach, namely exploiting convexity, \citet{ghodrati2021distribution} proposed to perform distribution-on-distribution regression via an optimal transport map. This approach yields a regression operator in a pointwise sense at the level of the original distributions leading to a clean and transparent interpretation. 
\citet{ghosal2023distributional} proposed a multidimensional extension of distributional outcome regression via quantile functions that handled both scalar and distributional predictors and can be seen as a generalization of \citet{yang2020random} and \citet{ghodrati2021distribution}.
Another recent example of multidimensional extension of Wasserstein distributional regression is the sliced Wasserstein regression proposed by \citet{chen2023sliced}.

The above-mentioned Wasserstein regression models focus on distributional representations using quantile functions. In contrast, our density-on-density regression approach offers an alternative that leverages the invariance property of the Fisher-Rao metric that results in an interpretable relationship between density functions. As such, we specify that both outcome and predictor densities are related via an isometric mapping which we call a warping function. The model is defined in the tangent space of half densities equipped with $\mathbb{L}^{2}$-metric. It is proposed to minimize the $\mathbb{L}^{2}$-distance between the half densities to estimate the warping function, which is equivalent to minimizing the Hellinger distance (or the Fisher-Rao metric) between outcome and predictor. Asymptotic properties of the estimator are studied under certain regularity conditions.

The proposed framework is motivated by an omics-imaging study from the Alzheimer's Disease Neuroimaging Initiative (ADNI). AD is an irreversible neurodegenerative disorder, mostly seen in the aging population. The disease process results in progressive declines in cognitive and behavioral function, especially memory, thereby having a broad impact on daily life and mortality. AD is quite prevalent, currently being ranked as the seventh-leading cause of death in the United States and the fifth-leading cause of death among Americans aged 65 and older~\citep{AD2022}. The precise causal mechanisms of AD are poorly understood, and no effective treatment is available. Understanding disease mechanisms and identifying therapeutic targets are thus crucially important. The ADNI study was launched in 2003 aiming to acquire assessments from various biological, clinical, and neuropsychological modalities during AD progression. This study focuses on two such modalities, cerebrospinal fluid (CSF) proteomics and brain structural magnetic resonance imaging (MRI). The CSF proteomics study aims to quantify protein (or protein segment) intensities in the CSF. The structural MRI offers information about brain structure, such as the volume of brain regions after applying a brain parcellation atlas. It has been shown that there exist connections between the deposition of protein markers, including amyloid-$\beta$ and tau, and the atrophy in certain brain areas, such as the entorhinal cortex and hippocampus~\citep{mormino2009episodic,pini2016brain,wesenhagen2020cerebrospinal}. Among the existing literature, more attention has been paid to the association between a single or a set of features from each modality, rather than the association between feature distributions.
In neuroimaging research, an example of considering densities as the observation unit is the distribution of brain functional connectivity acquired from the resting-state functional MRI experiments~\citep{petersen2016functional,tang2023differences}. 
Here, the abundance of various CSF proteins is assumed to follow a probability distribution and the acquired intensity data are random realizations, so the volume of different regions spanned over the whole brain. The study units are the density of the CSF protein abundance and the density of brain regional volumes from each individual. The objective then is to quantify and describe the association between two densities. Considering the volumetric density may seem unusual, by not employing the spatial information contained in the regions. However, considering image intensity histograms is a basic first step in nearly all image analysis from photographs to medical images of all sorts. In our case, if CSF-related volumetric atrophy is non-localized, or inconsistently localized across subjects, a density-based approach is more relevant than an approach that contrasts specific locations across subjects. Moreover, we reduce the multiple comparisons problem dramatically. Of course, in applied analysis, one would perform both sorts of analyses. We argue that histogram-based approaches (like density regression) are a natural first step, like an omnibus $F$-test in ANOVA, to be followed up with finer scale analyses. We further argue that the less frequent use of density-on-density methods in neuroimaging is primarily due to the much lower amount of methodological development in density regression when compared to high dimensional estimation and testing, or analyses following localized dimension reduction techniques, such as principal component analysis (PCA).

The rest of the manuscript is organized as follows. Section~\ref{sec:model} introduces the concept of a warping function, reviews necessary facts from the geometry of Riemannian manifolds, and proposes our density-on-density regression model. Section~\ref{sec:method} develops an estimator for the warping function and studies its asymptotic properties. In Section~\ref{sec:sim}, the performance of the proposed approach is evaluated and compared with existing methods through simulation studies. Section~\ref{sec:adni} applies the proposed density-on-density regression model to a proteomics-imaging study from ADNI. Section~\ref{sec:discussion} concludes the manuscript with a discussion.

\section{Model}
\label{sec:model}

Let $\mathcal{X}=\mathcal{Y}=[0,1]$ be the sample spaces. Here, without loss of generality, we assume the sample space is $[0,1]$ for both the predictor and response. Define the set of continuous probability density functions on $[0,1]$ as
\begin{equation}\label{eq:density_space}
  \mathscr{P}=\left\{f:[0,1]\mapsto\mathbb{R}_{\geq 0}~|~\int_{0}^{1}f(\omega)~\mathrm{d}\omega=1\right\}.
\end{equation}
Assume $(f_{i},g_{i})$ is a pair of density functions in $\mathscr{P}$ of two characteristics of unit $i$, for $i=1,\dots,n$, where $n$ is the number of units. 
Analogous to a linear (functional) regression problem, it is assumed that there exists a functional coefficient such that $g_{i}$ can be represented as a composition of $f_{i}$ and the coefficient function. However, for density functions, the coefficient function must satisfy the constraint that after composition, the resulting function is positive and has a unit integral.
The following definitions define such a coefficient function and an action that connects two densities.
\begin{definition}[Warping function]\label{def:warping}
  Let $\beta:[0,1]\rightarrow[0,1]$ be a function that satisfies the following: $\beta(0)=0$, $\beta(1)=1$, $\beta$ is invertible, and both $\beta$ and $\beta^{-1}$ are smooth. Then $\beta$ is called a \emph{boundary-preserving diffeomorphism} of $[0,1]$. Denote $\Gamma_{\beta}$ as the set of all such functions.
\end{definition}
\begin{definition}[An action $\odot$]\label{def:action}
Let $\mathbb{L}^{1}([0,1],\mathbb{R})$ denote the space of all (absolutely) integrable functions on $[0,1]$, we define the following action of $\Gamma_{\beta}$ on $\mathbb{L}^{1}([0,1],\mathbb{R})$:
\begin{equation}\label{eq:def_action}
  \mathbb{L}^{1}([0,1],\mathbb{R})\times\Gamma_{\beta}\rightarrow\mathbb{L}^{1}([0,1],\mathbb{R}), \quad f\odot\beta=(f\circ\beta)\beta',
\end{equation}
where $(f\circ\beta)=f(\beta(\omega))$ for $\omega\in[0,1]$ and $\beta'$ is the first-order derivative of $\beta$. 
\end{definition}
The action $\odot$ has the following properties.
\begin{enumerate}[(i)]
  \item Area preserving:
    \[
      \int (f\odot\beta)(\omega)~\mathrm{d}\omega=\int f(\beta(\omega))\beta'(\omega)~\mathrm{d}\omega=\int f(\omega)~\mathrm{d}\omega.
    \]
  \item Invertibility:
    \[
      (f\odot\beta)\odot\beta^{-1}=f.
    \]
  \item Grouping: for $\forall~\beta_{1},\beta_{2}\in\Gamma_{\beta}$,
    \[
      (f\odot\beta_{1})\odot\beta_{2}=f\odot(\beta_{1}\circ\beta_{2}).
    \]
\end{enumerate}
From the definition of $\beta$, it is monotonically increasing. Thus, if $f$ is a positive function, so is $f\odot\beta$. Also, $\odot$ is area preserving. 
Thus, $f\odot\beta$ is also a density function when $f$ is a density function. Analogous to a linear regression model, the invertibility property ensures the feasibility of swapping the response density and the predictor density. The grouping property, a form of associativity, can be considered as an equivalence to the change of variable in densities. It also corresponds to the grouping property in linear regression models. Setting $\beta_{2}=\beta_{1}^{-1}$, the grouping property infers invertibility.

The set of continuous probability density functions, $\mathscr{P}$, is a Banach manifold because the space of integrable functions on $[0,1]$ is a Banach space, but not a Hilbert space. Thus, one cannot directly define a functional regression model on density functions as on functions in a Hilbert space. Similar to the proposal in \citet{srivastava2007riemannian}, the Riemannian representation of density functions is utilized and the regression model is defined in the tied Riemannian manifold.
To define a geodesic path (and geodesic distance) between two density functions, as well as the Riemannian structure of $\mathscr{P}$, the following representation of a density function is considered. For $f\in\mathscr{P}$, let
\begin{equation}
  p(\omega)=\sqrt{f(\omega)}.
\end{equation}
$p(\omega)$ is also called the \textit{half density} of $f$, which falls in the positive orthant of the unit sphere, $\mathbb{S}_{\infty}$, defined as
\begin{equation}
  \mathbb{S}_{\infty}=\left\{p\in\mathbb{L}^{2}~|~\|p\|=\sqrt{\int p^{2}(\omega)~\mathrm{d}\omega}=1\right\}.
\end{equation}
$\mathbb{S}_{\infty}$ is a submanifold of $\mathbb{L}^{2}$, thus a Hilbert manifold. One can then define a distance between two densities through the geodesic distance between the corresponding half densities.

Before introducing the proposed density regression model, we first briefly review some concepts of Riemannian geometries. More details can be found in \citet{helgason2001differential,lang2012fundamentals,srivastava2016functional}. For an element $p\in\mathbb{S}_{\infty}$, the tangent space, denoted as $T_{p}(\mathbb{S}_{\infty})$, is defined as
\begin{equation}
  T_{p}(\mathbb{S}_{\infty})=\left\{v\in\mathbb{L}^{2}~|~\langle v,p\rangle =0\right\},
\end{equation}
where $\langle\cdot,\cdot\rangle$ is the inner product in $\mathbb{L}^{2}$. This inner product also makes $\mathbb{S}_{\infty}$ a \textit{Riemannian manifold} and defines the length of paths on the manifold. For a Hilbert manifold, the minimum length of paths between two points is achievable and the corresponding path is called a \textit{geodesic}. Considering an element (denoted as $p$) in a Riemannian manifold and a tangent vector (denoted as $v$) in its tangent space, there exists a unique constant-speed parameterized geodesic (denoted as $\alpha_{v}$) such that $\alpha_{v}(0)=p$ and $\alpha_{v}'(0)=v$. For the unit sphere $\mathbb{S}_{\infty}$ with $p\in\mathbb{S}_{\infty}$ and $v\in T_{p}(\mathbb{S}_{\infty})$, the geodesic has an explicit form and can be expressed as
\begin{equation}
  \alpha_{v}(\tau)(\omega)=\cos(\tau\|v\|)p(\omega)+\sin(\tau\|v\|)\frac{v(\omega)}{\|v\|}, \quad \omega\in[0,1],~\tau\in[0,1].
\end{equation}
An \textit{exponential map} is defined as $\mathrm{Exp}_{p}(v)=\alpha_{v}(1)$. For $\mathbb{S}_{\infty}$, the exponential map $\mathrm{Exp}_{p}:T_{p}(\mathbb{S}_{\infty})\rightarrow\mathbb{S}_{\infty}$ is given by
\begin{equation}
  \mathrm{Exp}_{p}(v)(\omega)=\cos(\|v\|)p(\omega)+\sin(\|v\|)\frac{v(\omega)}{\|v\|}, \quad \omega\in[0,1].
\end{equation}
This exponential map is many-to-one and surjective. 
If imposing a constraint that $\|v\|<\pi$, the map becomes one-to-one. The inverse of the exponential map, also called the \textit{logarithmic map} at $p$, has an analytical form. For $p,q\in\mathbb{S}_{\infty}$, $\mathrm{Exp}_{p}^{-1}:\mathbb{S}_{\infty}/\{-p\}\rightarrow T_{p}(\mathbb{S}_{\infty})$ is
\begin{equation}
  \mathrm{Log}_{p}(q)=\mathrm{Exp}_{p}^{-1}(q)(\omega)=\frac{\theta}{\sin\theta}\left(q(\omega)-p(\omega)\cos\theta\right), \quad \text{where } \theta=\cos^{-1}(\langle p,q\rangle) \text{ and } \omega\in[0,1].
\end{equation}

\subsection{A density-on-density regression model}

Consider continuous densities of two characteristics, $\{f_{i},g_{i}\}$, both defined on $[0,1]$, and a warping function $\beta\in\Gamma_{\beta}$, let
\begin{equation}
  p_{i}(\omega)=\sqrt{(f_{i}\odot \beta)(\omega)}=\sqrt{f_{i}(\beta(\omega))\beta'(\omega)}, \quad q_{i}(\omega)=\sqrt{g_{i}(\omega)}.
\end{equation}
The following data generating model is proposed for density-on-density regression,
\begin{equation}\label{eq:dens_model}
  q_{i}(\omega)=\mathrm{Exp}_{p_{i}}(e_{i})(\omega),
\end{equation}
where $e_{i}\in T_{p_{i}}(\mathbb{S}_{\infty})$ is a random variable in the tangent space of $p_{i}$. The $e_{i}$'s can be considered as the counterpart of the random errors in a regression model, but lie in different tangent spaces induced by the data points in $\mathbb{S}_{\infty}$. Assuming a base point $p_{0}\in\mathbb{S}_{\infty}$, for any point $p\in\mathbb{S}_{\infty}$, there exists a geodesic path to $p$ and a {\em parallel transport} map from $T_{p}(\mathbb{S}_{\infty})$ to $T_{p_{0}}(\mathbb{S}_{\infty})$. In geometry, parallel transport is a way of transporting geometric data along smooth curves in a manifold~\citep{spivak1970comprehensive}. 
The \textit{transported error} of $q_{i}$ with respect to $p_{i}$, denoted as $\varepsilon_{i}$, is then defined as the parallel transport of the actual error along the geodesic from $p_{i}$ to the base point $p_{0}$.  It is assumed that $\mathbb{E}(\varepsilon_{i})=0$ and $\mathrm{Var}(\varepsilon_{i})<\infty$. 
Figure~\ref{fig:geometry} presents a graphical demonstration of parallel transport. In the figure, the dashed line is the smooth geodesic path connecting $p_{0}$ and $p_{i}$. Along this path, the random error, $e_{i}\in T_{p_{i}}(\mathbb{S}_{\infty})$, with respect to $q_{i}$, is transported to a random error in $T_{p_{0}}(\mathbb{S}_{\infty})$, denoted as $\varepsilon_{i}$. By doing so, all $\varepsilon_{i}$'s are defined in the same tangent space, namely $T_{p_{0}}(\mathbb{S}_{\infty})$, and one can define properties of $\varepsilon_{i}$.
The idea of transporting model residuals was also employed in \citet{cornea2017regression}. 

\begin{figure}
  \begin{center}
    \includegraphics[width=0.5\textwidth, trim={1.75cm 0.75cm 1.75cm 0.5cm}, clip]{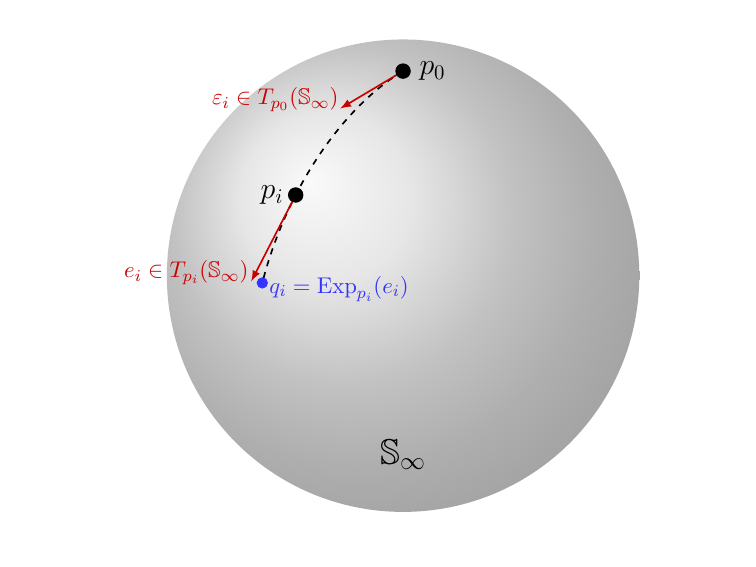}
  \end{center}
  \caption{\label{fig:geometry}Graphical demonstration of parallel transport.}
\end{figure}

Model~\eqref{eq:dens_model} is an extension of the function-to-function regression problem where both the predictor and outcome are (half) density functions. The model errors are defined in tangent spaces. From the definition of tangent space, this corresponds to the orthogonality assumption between the model error and the predictors in linear/functional regressions.
When the sample size $n=1$, the proposed model is equivalent to the problem of density registration~\citep{srivastava2016functional}. If the true warping function, $\beta^{*}$, is an identity function, the distribution of $X$ and $Y$ are identical (after centering and scaling to the sample space of $[0,1]$).

By modeling the densities via a warping function, Model~\eqref{eq:dens_model} offers an intuitive and straightforward presentation of how one distribution transforms into another.
Figure~\ref{fig:warping_example} shows examples of density functions before and after applying the warping function, $\beta(\omega)$. In the examples, $f(\omega)$ is the density function of the $\mathrm{Beta}(5,5)$ distribution, which is a symmetric function. For a convex shape of warping function, the resulted density function skews to the left (Figure~\ref{fig:warping_eg1}); while for a concave shape of warping function, the resulted density function skews to the right (Figure~\ref{fig:warping_eg2}). Figures~\ref{fig:warping_eg3} and~\ref{fig:warping_eg4} consider an ``S''-shape warping function. After warping, the variance of the new distribution decreases and the mode of the distribution shifts toward the inflection point of the warping function.

\begin{figure}
  \begin{center}
    \subfloat[\label{fig:warping_eg1}]{\includegraphics[width=0.5\textwidth,page=1]{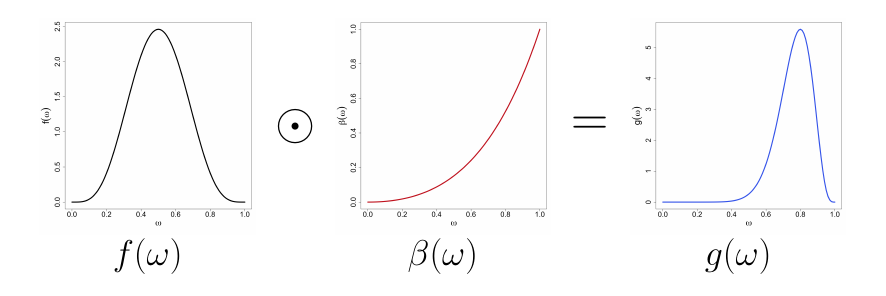}}
    \subfloat[\label{fig:warping_eg2}]{\includegraphics[width=0.5\textwidth,page=2]{Fig2.pdf}}

    \subfloat[\label{fig:warping_eg3}]{\includegraphics[width=0.5\textwidth,page=3]{Fig2.pdf}}
    \subfloat[\label{fig:warping_eg4}]{\includegraphics[width=0.5\textwidth,page=4]{Fig2.pdf}}
  \end{center}
  \caption{\label{fig:warping_example}Examples of density functions before and after applying the warping function, $\beta(\omega)$.}
\end{figure}

\section{Estimation Methods and Theory}
\label{sec:method}

Model~\eqref{eq:dens_model} is defined using the Riemannian representation of the density functions, that is the half densities. The space of the half densities is equipped with the $\mathbb{L}^{2}$-metric. Thus, it is proposed to estimate the warping function under the $\mathbb{L}^{2}$-metric.
For two sets of density functions, $\{f_{i}\}_{i=1}^{n}$ and $\{g_{i}\}_{i=1}^{n}$, under Model~\eqref{eq:dens_model}, the following estimator of the warping function is introduced,
\begin{equation}\label{eq:beta_est}
  \hat{\beta}=\underset{\beta\in\Gamma_{\beta}}{\arg\min}~\frac{1}{n}\sum_{i=1}^{n}\int_{0}^{1}\|\sqrt{g_{i}(\omega)}-\sqrt{(f_{i}\odot\beta)(\omega)}\|^{2}~\mathrm{d}\omega.
\end{equation}
The proposed estimator, $\hat{\beta}$, minimizes the average $\mathbb{L}^{2}$-distance between the two half densities across units. 
The reason of using half densities for optimization is that under the $\mathbb{L}^{2}$-metric, denoted as $d(\cdot,\cdot)$, for two densities $f_{1},f_{2}\in\mathscr{P}$, $d(f_{1},f_{2})\neq d(f_{1}\odot\beta,f_{2}\odot\beta)$, and thus can be arbitrarily close to each other \cite[the so-called pinching effect,][]{marron2015functional}. Using the square-root representation (or the half density), the following isometry property is satisfied.
\begin{lemma}\label{lemma:SRF_isometry}
  $f_{1}(\omega)$ and $f_{2}(\omega)$ are two density functions in $\mathscr{P}$. Let $p_{1}(\omega)=\sqrt{f_{1}(\omega)}$ and $p_{2}(\omega)=\sqrt{f_{2}(\omega)}$ denote the corresponding square-root functions (SRFs). Then,
  \begin{equation}
    d(p_{1},p_{2})=d((p_{1},\beta),(p_{2},\beta)),
  \end{equation}
  where $(p_{k},\beta)\equiv\sqrt{(f_{k}\odot\beta)(\omega)}=\sqrt{f_{k}(\beta(\omega))\beta'(\omega)}$, for $k=1,2$.
\end{lemma}
\noindent Based on this isometry property, the elastic distance between two density functions is defined via the distance between the corresponding half densities. The proposed estimator in~\eqref{eq:beta_est} thus aims to minimize an analogy of the elastic distance between two sets of densities in the sense of population average. Through half densities, it transforms the action on $\mathscr{P}$ to an action on $\mathbb{L}^{2}$. Existing statistical properties investigated in linear spaces, such as estimation consistency, can be then extended. 

It is noted that the optimization in~\eqref{eq:beta_est} is equivalent to minimizing the average Hellinger distance between the outcome density and the predictor density after warping, where the Hellinger distance between two density functions, denoted as $H(f_{1},f_{2})$, is defined as
\begin{equation}
  H^{2}(f_{1},f_{2})=\frac{1}{2}\int_{0}^{1}\left(\sqrt{f_{1}(\omega)}-\sqrt{f_{2}(\omega)}\right)^{2}~\mathrm{d}\omega=1-\int_{0}^{1}\sqrt{f_{1}(\omega)f_{2}(\omega)}~\mathrm{d}\omega.
\end{equation}
\begin{lemma}\label{lemma:Hellinger_dist}
  Considering two density functions, $f_{1}$ and $f_{2}$, a warping function $\beta\in\Gamma_{\beta}$, and the action defined in~\eqref{eq:def_action}, the Hellinger distance has the following properties.
  \begin{enumerate}[(i)]
    \item Invariance to simultaneous warping: for $\forall~\beta\in\Gamma_{\beta}$,
      \[
        H(f_{1},f_{2})=H(f_{1}\odot\beta,f_{2}\odot\beta).
      \]
    \item Consistency to random warpings: for $\forall~\beta_{1},\beta_{2}\in\Gamma_{\beta}$,
      \[
        \underset{\beta\in\Gamma_{\beta}}{\min}~H(f_{1},f_{2}\odot\beta)=\underset{\beta\in\Gamma_{\beta}}{\min}~H(f_{1}\odot\beta_{1},(f_{2}\odot\beta_{2})\odot\beta).
      \]
    \item Inverse symmetry (inverse consistency):
      \[
        \hat{\beta}=\underset{\beta\in\Gamma_{\beta}}{\arg\min}~H(f_{1},f_{2}\odot\beta) \quad\Rightarrow\quad \hat{\beta}^{-1}=\underset{\beta\in\Gamma_{\beta}}{\arg\min}~H(f_{1}\odot\beta,f_{2}).
      \]
  \end{enumerate}
\end{lemma}
\noindent The three properties in Lemma~\ref{lemma:Hellinger_dist} guarantee that the proposed estimator in~\eqref{eq:beta_est} or by minimizing the Hellinger distance yields a good solution of connecting two sets of densities. The first property, which is equivalent to the isometry property in Lemma~\ref{lemma:SRF_isometry}, is a fundamental property indicating that the action of~\eqref{eq:def_action} maintains the point-to-point correspondence between the two densities. And thus, $\beta$ is a diffeomorphism. Both Properties (ii) and (iii) are consequences of Property (i). Analogizing to the ordinary least square approach in a linear regression problem, Property (ii) is equivalent to the consistency of linear data transformation and Property (iii) corresponds to the symmetric property of swapping the independent and dependent variables. 
Given the equivalence of minimizing the average Hellinger distance and~\eqref{eq:beta_est}, in the following, the two formulations will be used interchangeably. The estimator in~\eqref{eq:beta_est} can be then rewritten as
\begin{equation}\label{eq:beta_estimator}
  \hat{\beta}=\underset{\beta\in\Gamma_{\beta}}{\arg\min}~\frac{1}{n}\sum_{i=1}^{n}H^{2}(g_{i},f_{i}\odot\beta).
\end{equation}

The optimization problem in~\eqref{eq:beta_estimator} is to minimize the $\mathbb{L}^{2}$-metric between half densities. To analogize to the ordinary least squares estimator in linear regression and demonstrate that $\hat{\beta}$ also minimizes the average distance between the outcome and warped densities under some metric, in the following, we first briefly review the connection between the Fisher-Rao metric for density functions and the $\mathbb{L}^{2}$-metric in $\mathbb{S}_{\infty}$, or more specifically in $\mathbb{S}_{\infty}^{+}$. For more details, readers can refer to \citet{friedrich1991fisher} and \citet{srivastava2016functional}. The Fisher-Rao metric was introduced to measure the distance between two distributions using a differential geometric approach~\citep{rao1945information}. It is a representation of the Riemannian metric in the space of probability distributions. When introducing the metric, the Fisher information matrix was used (thus is called the Fisher-Rao metric). For a parametric distribution family, the Fisher-Rao metric, therefore, offers a lower bound on the expected error associated with the estimators. Considering the nonparametric setting, for a density function $f\in\mathscr{P}$ and $f>0$, assume $v_{1},v_{2}$ are two vectors in the tangent space $T_{f}(\mathscr{P})$, the Fisher-Rao metric is defined as
\begin{equation}\label{eq:FRmetric}
  ((v_{1},v_{2}))_{f}=\int v_{1}(\omega)v_{2}(\omega)\frac{1}{f(\omega)}~\mathrm{d}\omega.
\end{equation}

For the space of $\mathbb{S}_{\infty}$, the explicit form of the geodesic is given in Section~\ref{sec:model}. The following gives the geodesic and the \textit{Riemannian distance} between two densities.
\begin{lemma}\label{prop:R-distance}
  For two continuous density functions $f_{1},f_{2}\in\mathscr{P}$ on the domain $[0,1]$, the Riemannian distance with respect to the Fisher-Rao metric (geodesic length distance) is
  \begin{equation}\label{eq:R-distance}
    d_{\mathrm{R}}(f_{1},f_{2})=\cos^{-1}\left(\int_{0}^{1}\sqrt{f_{1}(\omega)}\sqrt{f_{2}(\omega)}~\mathrm{d}\omega\right).
  \end{equation}
  The geodesic between $f_{1}$ and $f_{2}$ is
  \begin{equation}
    \alpha(\tau)=\left(\frac{\sin((1-\tau)\theta)}{\sin\theta}\sqrt{f_{1}}+\frac{\sin(\tau\theta)}{\sin\theta}\sqrt{f_{2}}\right)^{2},
  \end{equation}
  where $\theta=d_{\mathrm{R}}(f_{1},f_{2})$ and $\tau\in[0,1]$.
\end{lemma}
\noindent Additional properties of the Fisher-Rao metric are presented in Section~\ref{appendix:sub:FR-metric} of the supplementary materials. Section~\ref{appendix:sub:Rstructure_Gamma} studies the Riemannian structure of $\Gamma_{\beta}$.

Given the fact that the Fisher-Rao metric quantifies the information retained in the data, \citet{bauer2020diffeomorphic} defines the \textit{Optimal Information Transport (OIT)} as the diffeomorphism that minimizes the Riemannian distance between two densities, that is for $f_{1},f_{2}\in\mathscr{P}$, find a diffeomorphism, $\beta$, that minimizes $d_{\mathrm{R}}(f_{1},f_{2}\odot\beta)$. 
Proposition~\ref{prop:beta_Riemannian} below demonstrates that $\hat{\beta}$ in~\eqref{eq:beta_estimator} also minimizes the average Riemannian distance between the outcome and warped densities. In this sense, $\hat{\beta}$ is an optimal information transport across units. 
\begin{proposition}\label{prop:beta_Riemannian}
  For a fixed $n$, assume that $\hat{\beta}$ is the solution in~\eqref{eq:beta_estimator}. Then, $\hat{\beta}$ is also the solution to the following minimization problem,
  \begin{equation}
    \min_{\beta\in\Gamma_{\beta}}~\frac{1}{n}\sum_{i=1}^{n}d_{\mathrm{R}}(g_{i},f_{i}\odot\beta).
  \end{equation} 
\end{proposition}

The following theorem shows that $\hat{\beta}$ is a consistent estimator of $\beta$. 
In practice, neither $f_{i}$ nor $g_{i}$ is directly observed but can be estimated from the observed (discretized) data. When the considered estimator of the densities is consistent, such as the kernel density estimator~\citep{wasserman2006all}, the consistency of estimating $\beta$ follows.
\begin{theorem}\label{thm:beta_asmp}
  Assume $\{f_{i}\}$ and $\{g_{i}\}$ are continuous density functions on $[0,1]$ satisfying~\eqref{eq:dens_model}, where $p_{i}=\sqrt{f_{i}\odot\beta}$ and $q_{i}=\sqrt{g_{i}}$, for $i=1,\dots,n$. Denote $\varepsilon_{i}$ as the parallel transported model error with $\mathbb{E}(\varepsilon_{i})=0$ and $\mathrm{Var}(\varepsilon_{i})<\infty$. Let $\hat{\beta}$ be an estimator of $\beta$ that minimizes the average Hellinger distance (as in~\eqref{eq:beta_estimator}). $\hat{\beta}$ is a consistent estimator of $\beta$, that is,
    \begin{equation}
      d(\hat{\beta},\beta)\overset{\mathcal{P}}{\longrightarrow} 0, \quad \text{as } n\rightarrow\infty,
    \end{equation}
  where $d(\cdot,\cdot)$ is a distance metric in $\Gamma_{\beta}$ and $\overset{\mathcal{P}}{\longrightarrow}$ denotes convergence in probability.
\end{theorem}

\subsection{Algorithm}
\label{sub:algorithm}

The space of warping functions, $\Gamma_{\beta}$, is not a linear space. Thus, one cannot estimate $\beta$ via kernel representations as what is widely used for functional regressions in the reproducing kernel Hilbert space.
To estimate a warping function in the curve registration problem, \citet{ramsay1998curve} proposed to estimate through the smooth monotone transformations. In this study, the same idea is employed and a similar estimating approach is proposed. In addition to being strictly increasing, it is assumed that the warping function, $\beta$, has an integrable second-order derivative. Then $\beta$ can be described by the following homogeneous linear differential equation,
\begin{equation}\label{eq:w_def}
  \mathcal{D}^{2}\beta=w~\mathcal{D}\beta,
\end{equation}
where $\mathcal{D}$ is the derivative operator, $w$ is the weight function. Under the boundary conditions subject to $\beta(0)=0$ and $\beta(1)=1$, the solution is
\begin{equation}\label{eq:beta_homosolution}
  \beta(\omega)=C_{1}\int_{0}^{\omega}\exp\{W(s)\}~\mathrm{d}s,
\end{equation}
where $W(s)=\int_{0}^{s}w(u)~\mathrm{d}u$ and $C_{1}=\left(\int_{0}^{1}\exp\{W(s)\}~\mathrm{d}s\right)^{-1}$.
Assuming $w(\omega)$ is in a reproducing kernel Hilbert space, a penalized estimator of $w$ is proposed as
\begin{equation}
  \hat{w}(\omega)=\underset{w}{\arg\min}~\frac{1}{n}\sum_{i=1}^{n}H^{2}(g_{i},f_{i}\odot\beta)+\lambda\int w^{2}(\omega)~\mathrm{d}\omega,
\end{equation}
where by the Representer Theorem, $w(\omega)$ can be represented by a finite linear combination of reproducing kernel functions $\{\phi_{k}(\omega)\}$ on $[0,1]$,
\begin{equation}
  w(\omega)=\sum_{k=0}^{K}\alpha_{k}\phi_{k}(\omega),
\end{equation}
and $\alpha_{k}$'s are the coefficients. 

The estimation procedure is summarized in Algorithm~\ref{alg:beta_opt}. In Step 1, two tuning parameters, $K$ and $\lambda$, are pre-specified. To choose the values, procedures such as cross-validation can be employed. For $\lambda$, through multiple applications, it is found that values between $10^{-4}$ and $10^{-2}$ can attain a good performance~\citep{ramsay1998curve}. Step 2 yields the estimate of the warping function using the formulation in~\eqref{eq:beta_homosolution}.

\begin{algorithm}
  \caption{\label{alg:beta_opt}Estimate the warping function, $\beta$, via the smooth monotone transformation.}
  \begin{description}
    \item[Step 0] Let $\{f_{i}\}_{i=1}^{n}$ and $\{g_{i}\}_{i=1}^{n}$ denote the predictor and outcome density functions, respectively.
    \item[Step 1] With a set of chosen basis functions, a choice of the number of basis ($K$), and the tuning parameter $\lambda$, estimate $\{\alpha_{k}\}$ by solving the following optimization problem:
      \begin{equation}
        \underset{\{\alpha_{k}\}}{\text{minimize}}~\frac{1}{2n}\sum_{i=1}^{n}\int\left(\sqrt{g_{i}(\omega)}-\sqrt{f_{i}(\beta(\omega))\beta'(\omega)}\right)^{2}~\mathrm{d}\omega+\lambda\int w^{2}(\omega)~\mathrm{d}\omega,
      \end{equation}
    where $\beta(\omega)$ is represented by $w(\omega)$ via~\eqref{eq:beta_homosolution}.
    \item[Step 2] With the estimate, $\{\hat{\alpha}_{k}\}$, one can obtain the estimate of $\beta(\omega)$ as
      \begin{equation}
        \hat{\beta}(\omega)=\hat{C}_{1}\int_{0}^{\omega}\prod_{k}\exp\left\{\hat{\alpha}_{k}\int_{0}^{s}\phi_{k}(u)~\mathrm{d}u\right\}~\mathrm{d}s,
      \end{equation}
      where $\hat{C}_{1}=\left(\int_{0}^{1}\exp\{\hat{W}(s)\}~\mathrm{d}s\right)^{-1}$, $\hat{W}(s)=\int_{0}^{s}\hat{w}(u)~\mathrm{d}u$, and $\hat{w}(\omega)=\sum_{k=0}^{K}\hat{\alpha}_{k}\phi_{k}(\omega)$.
  \end{description}
\end{algorithm}

\subsection{Inference}
\label{sub:inference}

In this section, we suggest an inference strategy based on the algorithm introduced in Section~\ref{sub:algorithm}. As discussed in Section~\ref{appendix:sub:Rstructure_Gamma} of the supplementary materials, the space of the warping function, $\Gamma_{\beta}$, is not a linear space. Thus, we first focus on the asymptotic properties of the weight function, $w$, defined in~\eqref{eq:w_def}.

\begin{theorem}\label{thm:w_asmp}
  Let $f_{i}$ and $g_{i}$ be the predictor and outcome densities, respectively, for $i=1,\dots,n$. Denote $\hat{w}$ as the estimator of $w$ that minimizes the Hellinger distance defined in~\eqref{eq:beta_estimator} using representation~\eqref{eq:beta_homosolution}. Under Assumptions (A1)--(A3) (in Section~\ref{appendix:sub:proof_thm_w_asmp} of the supplementary materials), for $\omega\in[0,1]$,
  \begin{equation}
    \frac{\hat{w}(\omega)-w(\omega)}{\sqrt{C_{n}(\omega)}}\overset{\mathcal{D}}{\longrightarrow}\mathcal{N}(0,1),
  \end{equation}
  where
  \[
    C_{n}(\omega)=\frac{1}{n}A_{n}^{-1}(\omega)B_{n}(\omega)A_{n}^{-1}(\omega), \quad A_{n}(\omega)=\frac{1}{n}\sum_{i=1}^{n}\frac{\partial^{2} L_{i}}{\partial w^{2}}(\omega), \quad B_{n}(\omega)=\frac{1}{n}\sum_{i=1}^{n}\mathrm{Var}\left(\frac{\partial L_{i}}{\partial w}(\omega)\right),
  \]
  \[
    L_{i}(\beta)=\frac{1}{2}\int\left(\sqrt{g_{i}(\omega)}-\sqrt{f_{i}(\beta(\omega))\beta'(\omega)}\right)^{2}~\mathrm{d}\omega,
  \]
  and $\overset{\mathcal{D}}{\longrightarrow}$ denotes convergence in distribution.
\end{theorem}

Theorem~\ref{thm:w_asmp} derives the pointwise asymptotic distribution of the estimator of the $w$ function. 
The asymptotic distribution is derived based on the theoretical results of an $M$-estimator extended to functional data~\citep{cox1983asymptotics}. For the warping function, $\beta(\omega)$, a pointwise confidence interval can be then constructed using the formula in~\eqref{eq:beta_homosolution}.


\section{Simulation Study}
\label{sec:sim}

In this section, the performance of the proposed density-on-density regression model is evaluated via simulation studies. In the studies, the density function of $\mathrm{Beta}(2,5)$ distribution is set to be the predictor density ($f_{i}$). A convex shape of $\beta$ function (similar to $\beta(\omega)$ in Figure~\ref{fig:warping_eg1}) is considered. The outcome density function ($g_{i}$) is then generated following model~\eqref{eq:dens_model}, where the error function is generated from the tangent function of $p_{i}=\sqrt{f_{i}\odot\beta}$ multiplied by a multiplier generated from a uniform distribution with mean zero. 
Figure~\ref{fig:sim_data} shows the density function of $f_{i}(\omega)$ ($\mathrm{Beta}(2,5)$) and the generated density functions of $g_{i}(\omega)$ for $n=100$ random samples.
Four methods are considered to compare the performance. (1) The proposed density-on-density regression using the true density function, denoted as DoDR-True. (2) The proposed density-on-density regression using the kernel estimator of the density function from observations generated from the density functions, denoted as DoDR-Est. (3) The Wasserstein regression introduced by \citet{chen2021wassersteinreg}, denoted as CLM. (4) A distribution-on-distribution regression approach via optimal transport maps by \citet{ghodrati2021distribution}, denoted as GP. 
In (1) and (2), a $B$-spline fitting is considered with $K=4$ basis functions. The tuning parameter, $\lambda$, is chosen based on $5$-fold cross-validation.
For (1) and (2), an estimate of the warping function, $\beta(\omega)$, is obtained. To evaluate the performance, the distance between the estimate and the truth is calculated from the square-root slope transformation (see a discussion in Section~\ref{appendix:sub:Rstructure_Gamma} of the supplementary materials). For all four approaches, a fitted density function can be acquired for each sample. The Hellinger distance between the density function and the fitted density function is then utilized to compare the proposed approaches to approaches (3) and (4).
Multiple sample size combinations are considered with $n=50,100,500$ and $m_{i1}=m_{i2}=50,100,500$, where $m_{i1}$ and $m_{i2}$ are the number of observations generated from $f_{i}$ and $g_{i}$, respectively, in approach (2). Simulations are repeated for $200$ replications.

TThe performance of the proposed approach, (1) and (2), is first examined. Using the true density functions, Figure~\ref{fig:sim_betaEst_true} presents the estimated $\beta(\omega)$ function from $n=100$ samples. 
The estimate (black solid line) is very close to the truth (red dashed line). In practice, the density functions need to be estimated from the observed data first. Figure~\ref{fig:sim_data_est} presents the estimated density functions from $m_{i1}=m_{i2}=100$ observations and Figure~\ref{fig:sim_betaEst_est} presents the estimated $\beta(\omega)$ function using the estimated densities. Compared to the results from the true density functions, the estimation variation is higher. In both Figure~\ref{fig:sim_data} and Figure~\ref{fig:sim_data_est}, the red dashed lines are the fitted density functions using the estimated warping function. Via warping, the proposed approach shifts the center and shape of the density functions. Table~\ref{table:sim_beta} presents the distance between the estimate and the true $\beta(\omega)$ function. As the number of observations ($m_{i1},m_{i2}$) and the number of samples ($n$) increase, the distance and standard error of DoDR-Est decrease. For DoDR-True, as $n$ increases, the performance improves with lower distance and standard error; while as $m_{i1}$ and $m_{i2}$ increase, the performance almost remains the same. This is expected as the true density functions are used for estimation. 
Table~\ref{table:sim_pred} presents the average Hellinger distance between the outcome density function and the fitted density function with a sample size of $n=100$ and $m_{i1}=m_{i2}=100$. The proposed DoDR-True yields the lowest distance followed by DoDR-Est. The average distance from CLM and GP approaches is much higher. 
The CLM approach was designed to perform regression under the Wasserstein metric and the GP approach made an extension by replacing the notion of expectation with a Wasserstein-Fr{\'e}chet mean based on the theory of optimal transport. However, under the Wasserstein metric, the isometry of warping density functions does not hold (see Section~\ref{appendix:sub:Wasserstein} of the supplementary materials) leading to less desired performance when data are generated from the proposed model.
Here, it should be noted that both CLM and GP assumed different data generating mechanisms from the proposal~\eqref{eq:dens_model} and the objective is to construct regressions to guarantee closeness in the sense of Wasserstein distance. Thus, these two approaches are not directly comparable to ours. The proposed approach offers an alternative for distributional regression and shall not be interpreted as an improvement over CLM and GP.

\begin{figure}
  \begin{center}
    \subfloat[\label{fig:sim_data}]{\includegraphics[width=0.45\textwidth]{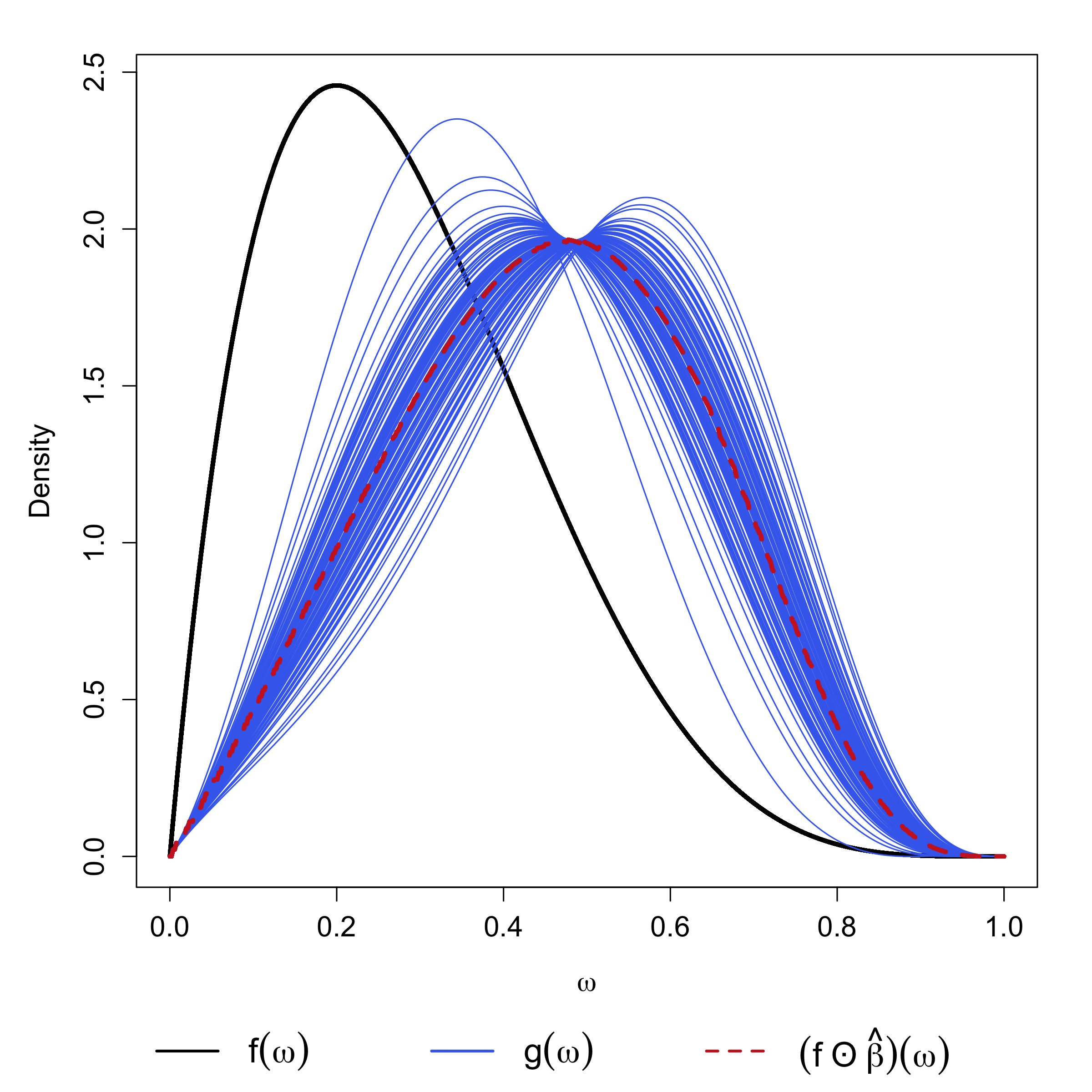}}
    \enskip{}
    \subfloat[\label{fig:sim_betaEst_true}]{\includegraphics[width=0.45\textwidth]{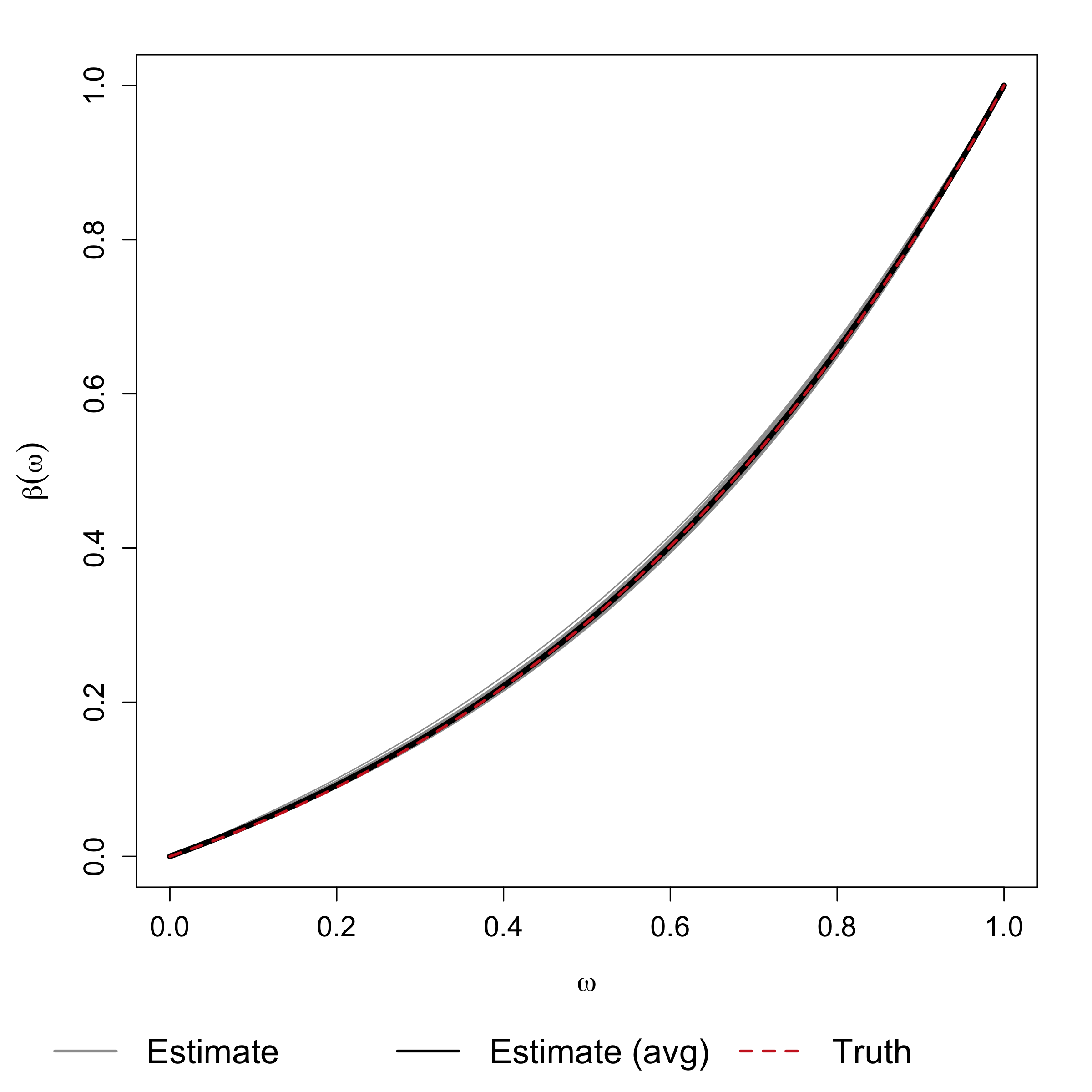}}

    \subfloat[\label{fig:sim_data_est}]{\includegraphics[width=0.45\textwidth]{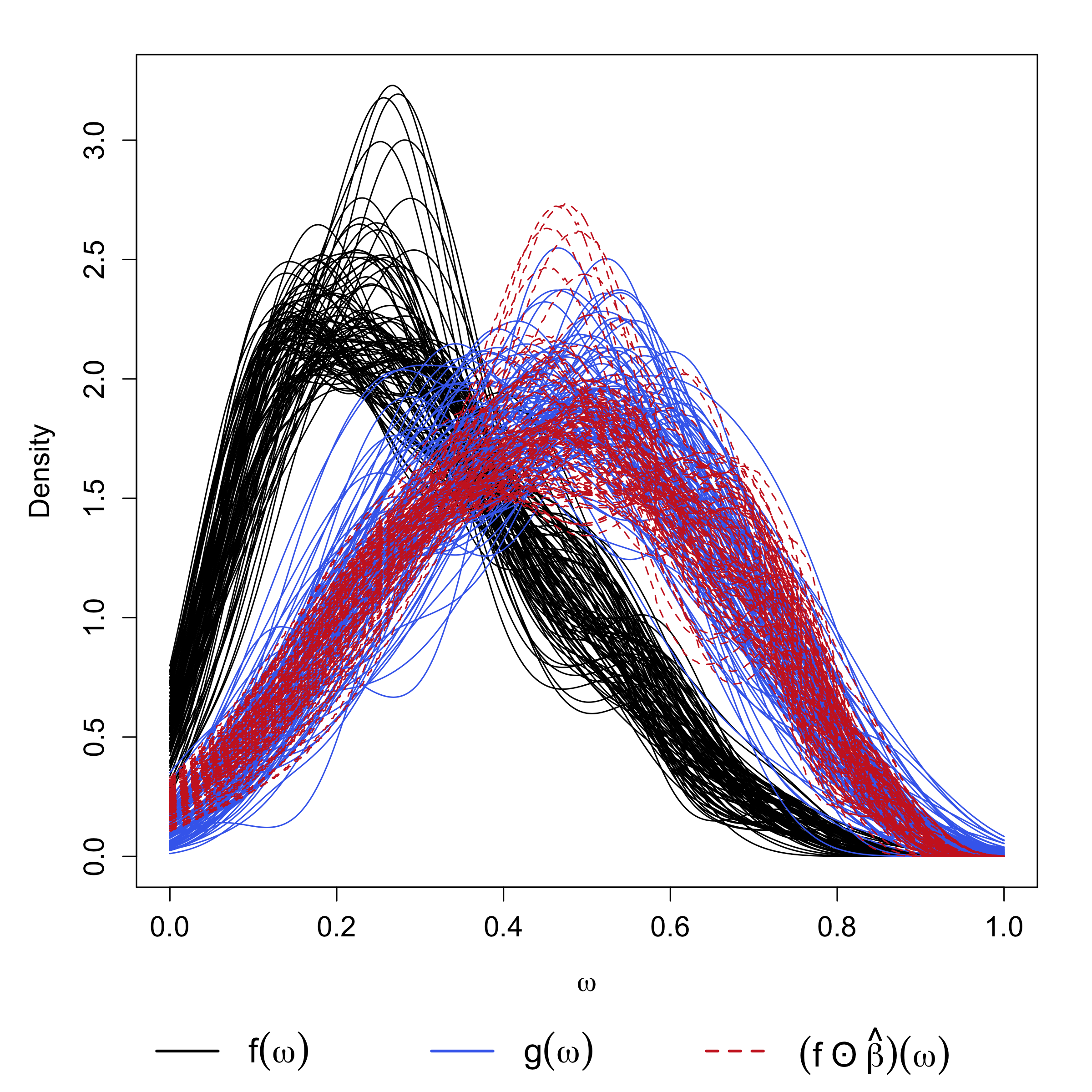}}
    \enskip{}
    \subfloat[\label{fig:sim_betaEst_est}]{\includegraphics[width=0.45\textwidth]{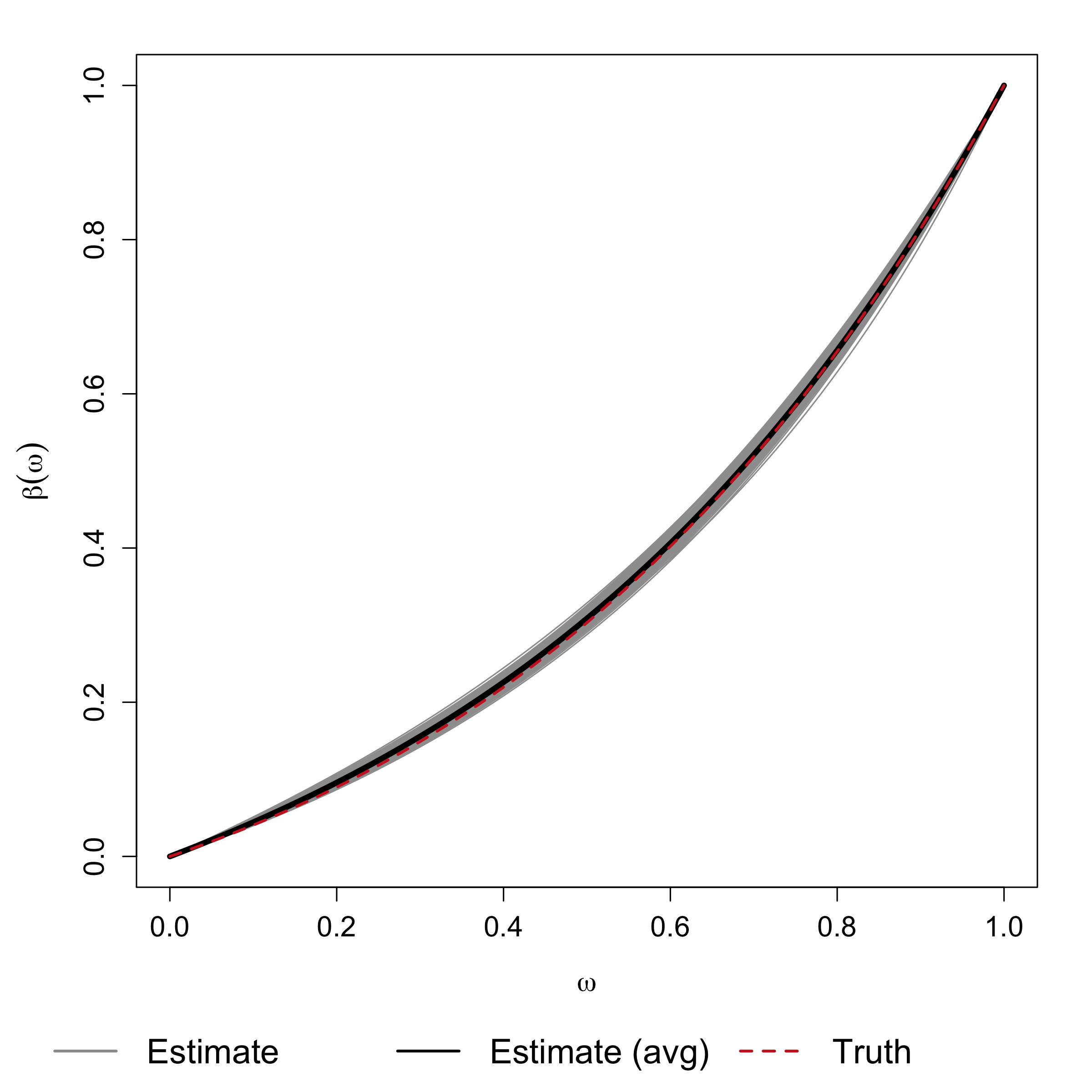}}
  \end{center}
  \caption{\label{fig:sim_data_dens}(a) Density functions $f$ and $g$ and the fitted density function using the estimated $\beta(\omega)$ from DoDR-True with $n=100$ samples. (b) Estimated $\beta(\omega)$ function using the true density functions (DoDR-True) of $200$ replications with a sample size of $n=100$. (c) Estimated density functions of $f$ and $g$ from $m_{i1}=m_{i2}=100$ observations and the fitted density functions using the estimated $\beta(\omega)$ from DoDR-Est with $n=100$ samples. (d) Estimated $\beta(\omega)$ function using the estimated density functions (DoDR-Est) of $200$ replications with a sample size of $n=100$ and $m_{i1}=m_{i2}=100$.}
\end{figure}
\begin{table}
  \caption{\label{table:sim_beta}Average distance between the estimated $\beta$ and the truth ($d(\hat{\beta},\beta)$) and the standard error (SE) over $200$ simulations in the simulation study.}
  \begin{center}
    \begin{tabular}{l l r r c r r c r r}
      \hline
      & & \multicolumn{2}{c}{$m_{i1}=m_{i2}=50$} && \multicolumn{2}{c}{$m_{i1}=m_{i2}=100$} && \multicolumn{2}{c}{$m_{i1}=m_{i2}=500$} \\
      \cline{3-4}\cline{6-7}\cline{9-10}
      & & \multicolumn{1}{c}{$d(\hat{\beta},\beta)$} & \multicolumn{1}{c}{SE} && \multicolumn{1}{c}{$d(\hat{\beta},\beta)$} & \multicolumn{1}{c}{SE} && \multicolumn{1}{c}{$d(\hat{\beta},\beta)$} & \multicolumn{1}{c}{SE} \\
      \hline
      & DR-True & $0.009$ & $0.005$ && $0.008$ & $0.004$ && $0.009$ & $0.004$ \\
      \multirow{-2}{*}{$n=50$} & DR-Est & $0.020$ & $0.010$ && $0.018$ & $0.008$ && $0.014$ & $0.005$ \\
      \hline
      & DR-True & $0.008$ & $0.004$ && $0.008$ & $0.006$ && $0.008$ & $0.005$ \\
      \multirow{-2}{*}{$n=100$} & DR-Est & $0.016$ & $0.007$ && $0.015$ & $0.006$ && $0.013$ & $0.005$ \\
      \hline
      & DR-True & $0.008$ & $0.005$ && $0.008$ & $0.004$ && $0.008$ & $0.005$ \\
      \multirow{-2}{*}{$n=500$} & DR-Est & $0.014$ & $0.006$ && $0.012$ & $0.005$ && $0.010$ & $0.004$ \\
      \hline
    \end{tabular}
  \end{center}
\end{table}
\begin{table}
  \caption{\label{table:sim_pred}Average Hellinger distance between the outcome density function ($g_{i}$) and the predicted density function ($\hat{g}_{i}=f_{i}\odot \hat{\beta})$ with a sample size of $n=100$ and $m_{i1}=m_{i2}=100$ in the simulation study. The calculation is an average over $200$ simulations. SE: standard error.}
  \begin{center}
    \begin{tabular}{l r r r r}
      \hline
      & \multicolumn{1}{c}{DR-True} & \multicolumn{1}{c}{DR-Est} & \multicolumn{1}{c}{CLM} & \multicolumn{1}{c}{GP} \\
      \hline
      $\bar{H}(g_{i},\hat{g}_{i})$ (SE) & $0.041$ $(0.003)$ & $0.099$ $(0.004)$ & $0.364$ $(0.003)$ & $0.775$ $(0.004)$ \\
      \hline
    \end{tabular}
  \end{center}
\end{table}


\section{The Alzheimer's Disease Neuroimaging Initiative Study}
\label{sec:adni}

We apply the proposed approach to data collected by the Alzheimer's Disease Neuroimaging Initiative (ADNI, \url{adni.loni.usc.edu}). 
The ADNI study was launched in 2003 as a public-private partnership. The primary goal is to test whether serial magnetic resonance imaging (MRI), positron emission tomography (PET), other biological markers, and clinical and neuropsychological assessments can be combined to measure the progression of mild cognitive impairment (MCI) and early AD. 
With data collected from different biological modalities, it enables the investigation of underlying complex interrelated mechanisms. 
As discussed in Section~\ref{sec:intro}, the proposed approach is implemented to investigate the association between the intensity distribution ($f_{i}$) of $m_{i1}=320$ peptides annotated from $142$ proteins in the cerebrospinal fluid (CSF) and the volume distribution ($g_{i}$) of $m_{i2}=135$ brain regions of interest (ROIs).

The CSF proteomics data were acquired using the technique of targeted liquid chromatography multiple reaction monitoring mass spectrometry. Based on the existing knowledge of AD, a list of protein fragments (or peptides) was sent to the detector. After processing steps of peak integration, outliers detection, normalization, quantification, and quality control using test/retest samples, intensities of the $320$ peptides were obtained for each unit. Proper data transformation, such as logarithmic transformation, was applied for data analysis. The brain imaging data were acquired using anatomical MRI. Following a standard pipeline, images were preprocessed and mapped to an atlas of $135$ ROIs spanning the entire brain to extract the volumetric measures~\citep{doshi2016muse}. Before performing analysis, the volume of each region was normalized by the total intracranial volume to remove the effect due to the variation in individual brain size.
We apply the proposed density-on-density regression on cognitive normal (CN) subjects ($n=86$), subjects diagnosed with MCI ($n=135$) and AD ($n=66$), separately. Both the proteomics and volumetric data are recentered and rescaled to the interval of $[0,1]$. 

Figure~\ref{fig:adni_beta} presents the estimated warping function and the $95\%$ point-wise confidence interval for each diagnostic group. For CN subjects, the estimated warping function is very close to the identity function with a slight concave curvature suggesting that the density function of the brain volumes slightly skews to the right compared to the density function of peptide intensities in the CSF. For MCI subjects, the $95\%$ point-wise confidence interval of the estimated warping function covers the identify function with $\omega\in(0.2,0.7)$. At the tail values of both sides, the estimated warping function deviates from identity toward right-skewness. The estimated warping function from AD subjects yields the highest concave curvature suggesting the largest divergence between the two distributions.
For each diagnostic group, one subject is chosen and data distributions, as well as the predicted distribution of brain volumes, are plotted in Figure~\ref{fig:adni_density}. From the figures, the estimated warping function aligns the distribution of peptide intensities to the distribution of brain volumes. Among the healthy aging population, the two distributions are very much like each other with the volumetric distribution slightly skewed to the right. Among subjects diagnosed with MCI or AD, the skewness of the volumetric distribution is greater. 
As a prodromal stage of AD, atrophy in the medial temporal lobe, including the hippocampus and entorhinal cortex, has been consistently observed in MCI. However, it only accounts for a relatively small portion of whole brain loss~\citep{tabatabaei2015cerebral}. When developed AD, topographical progression of cortical atrophy has been observed following a temporal--parietal--frontal trajectory while motor areas until late stages of the disease~\citep{pini2016brain}. Thus, a distribution with greater skewness is observed among AD compared to MCI.
Current existing analytical approaches focus on the identification of protein/volumetric markers for AD and the investigation of the associations between the markers from the two modalities. The proposed approach offers a way of studying the association between the distribution densities of these two types of biological assessments.

\begin{figure}
  \begin{center}
    \includegraphics[width=0.6\textwidth]{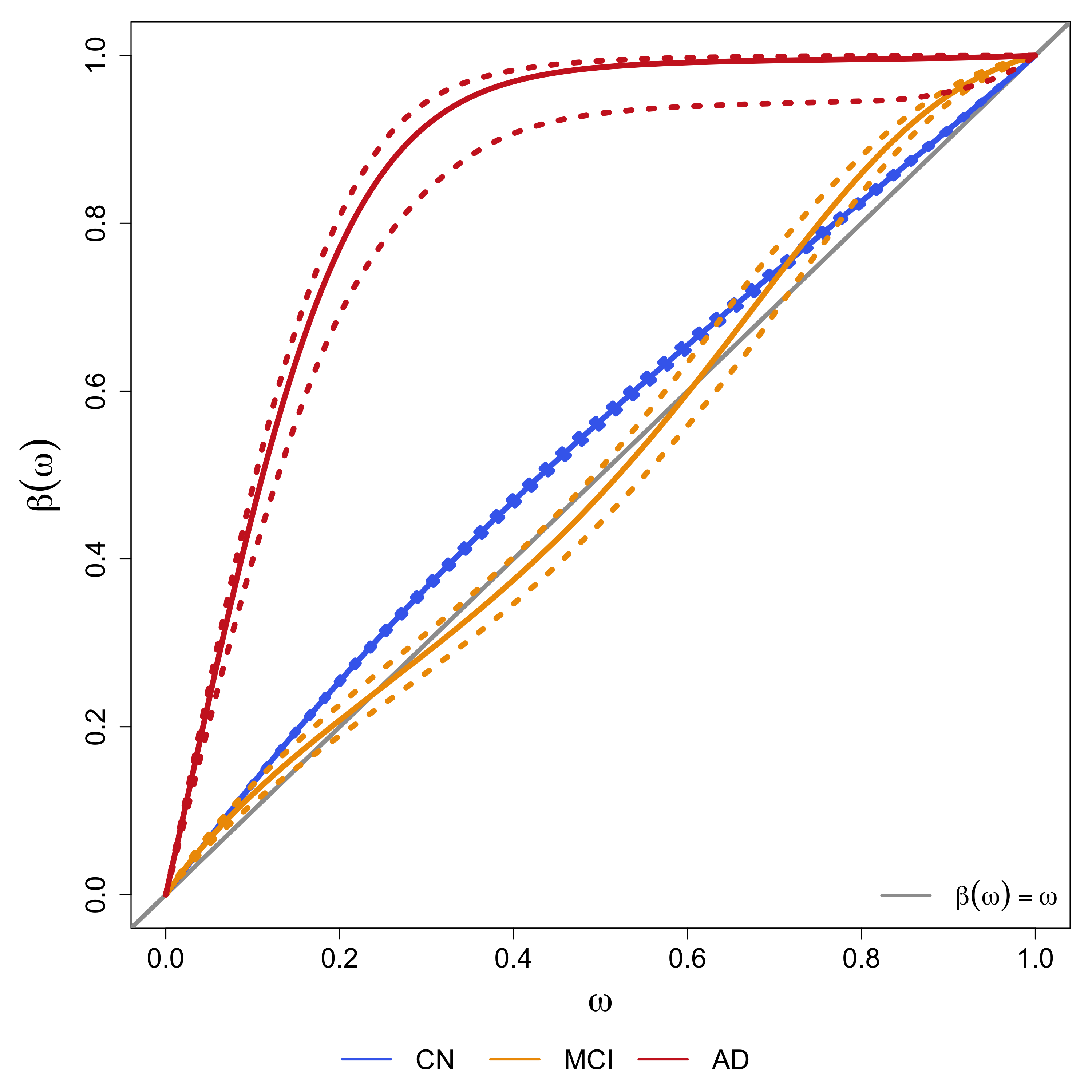}
  \end{center}
  \caption{\label{fig:adni_beta}Estimated warning function (solid lines) and the $95\%$ point-wise confidence interval (dashed lines) in the ADNI proteomic-imaging study. The estimation is conducted for CN (blue), MCI (yellow), and AD (red), separately.}
\end{figure}
\begin{figure}
  \begin{center}
    \subfloat[\label{subfig:adni_dens_CN}CN]{\includegraphics[width=0.3\textwidth]{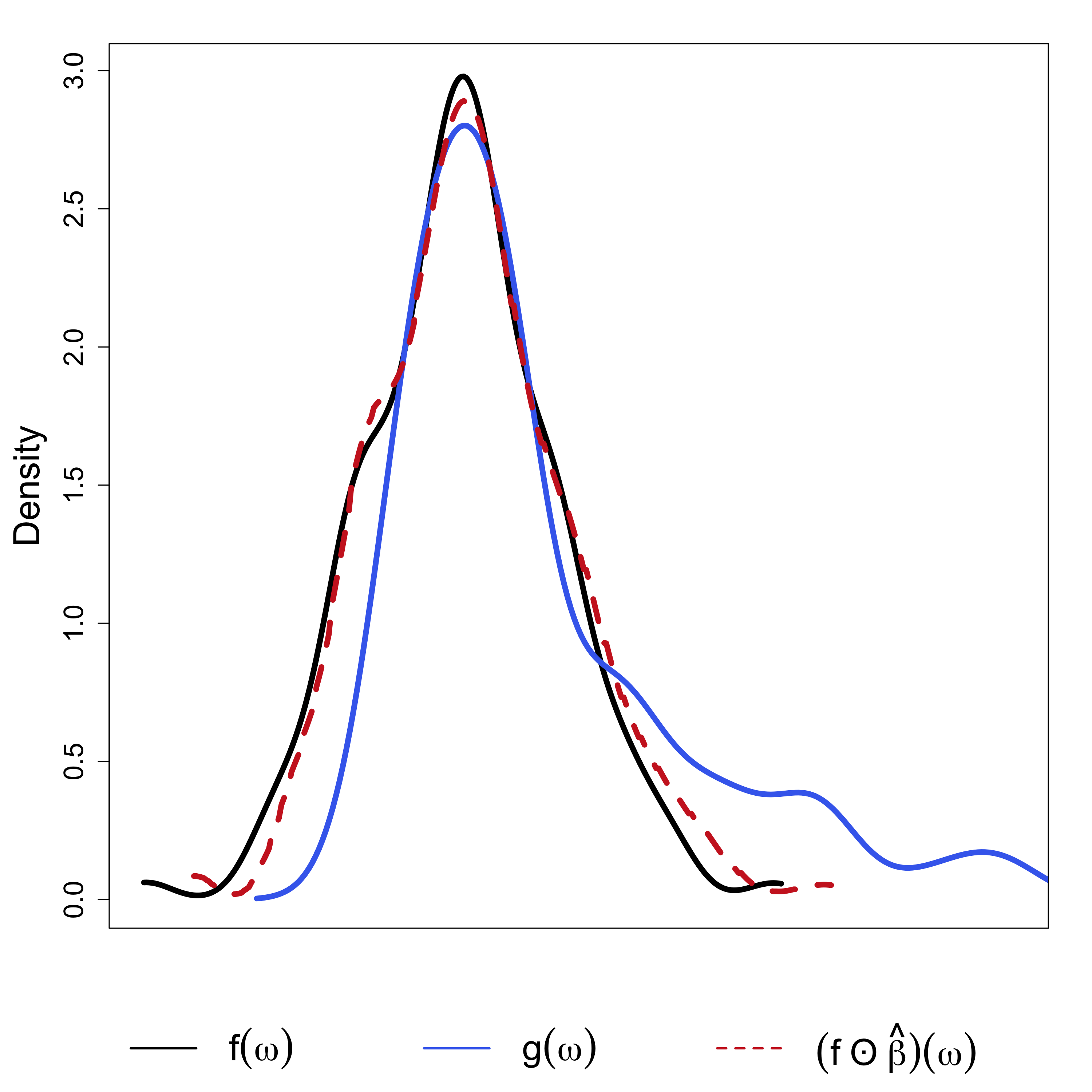}}
    \enskip{}
    \subfloat[\label{subfig:adni_dens_MCI}MCI]{\includegraphics[width=0.3\textwidth]{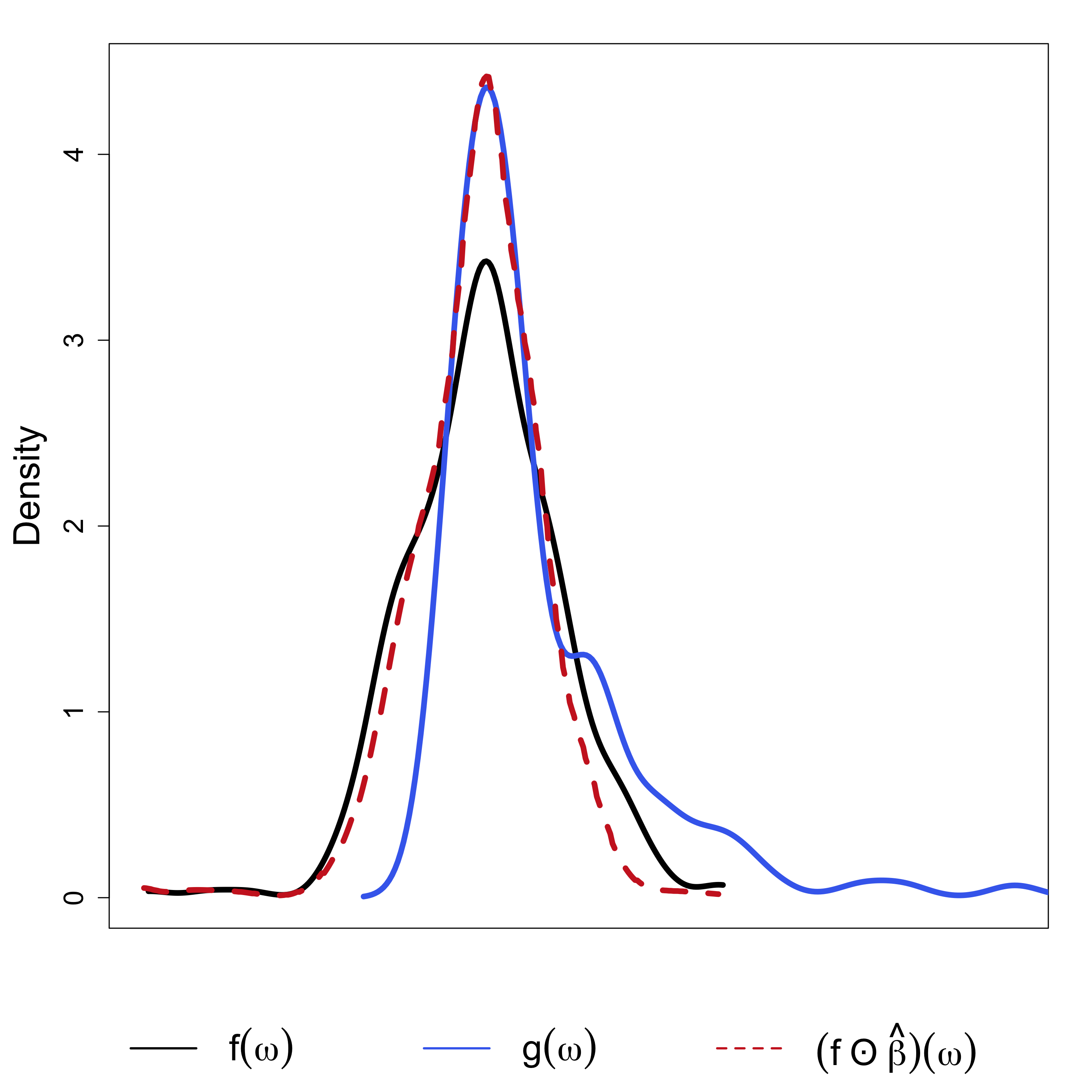}}
    \enskip{}
    \subfloat[\label{subfig:adni_dens_AD}AD]{\includegraphics[width=0.3\textwidth]{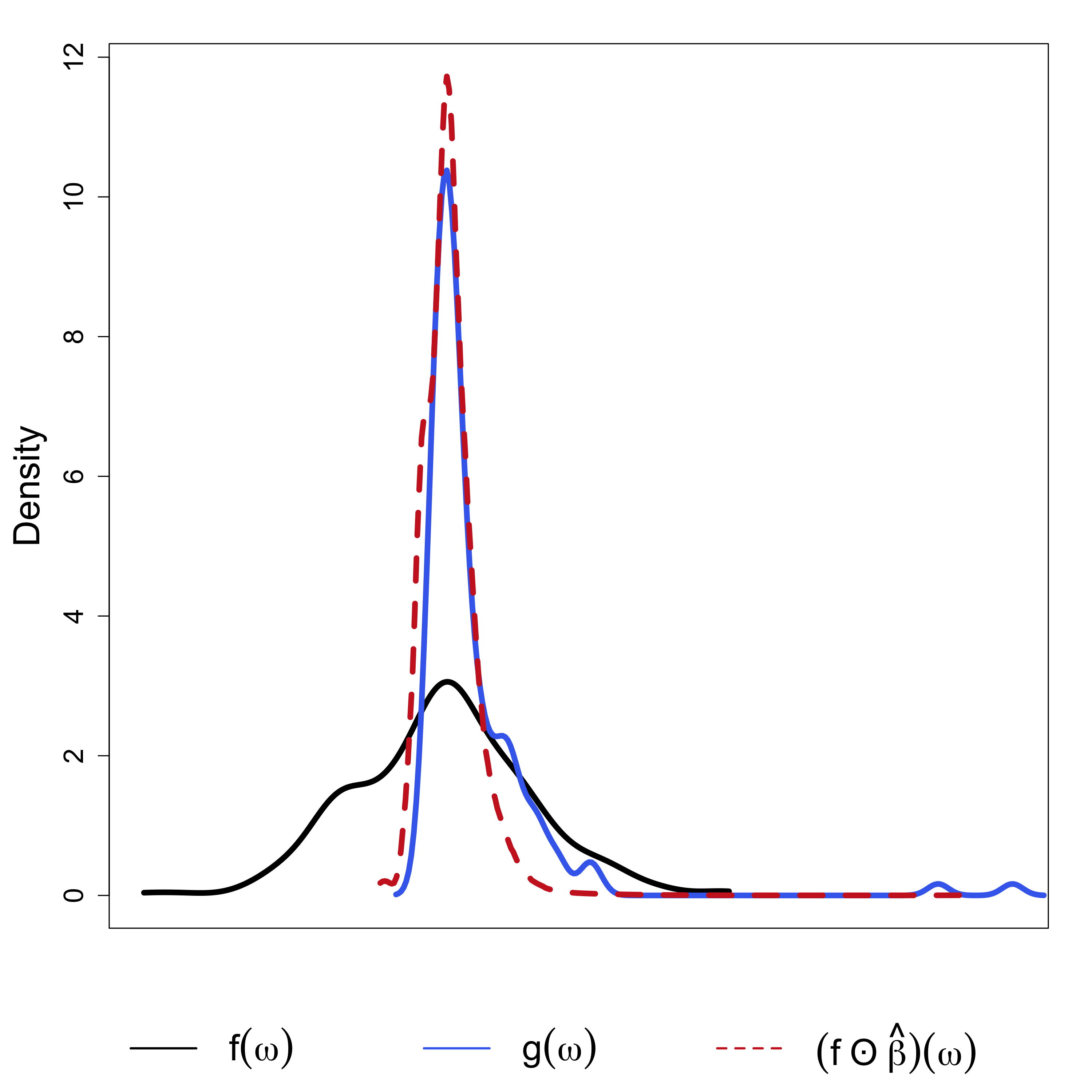}}
  \end{center}
  \caption{\label{fig:adni_density}Density function of the proteomics data ($f(\omega)$, black solid line), brain volumetric data ($g(\omega)$, blue solid line), and fitted density of brain volumetric data ($\hat{g}(\omega)$, red dashed line).}
\end{figure}

\section{Discussion}
\label{sec:discussion}

In this study, a density-on-density regression model is introduced, where the association between densities is elucidated via a warping function. The proposed model offers an alternative for distributional regression and has the advantage of a straightforward demonstration of how one density transforms into another. Using the Riemannian representation of density functions, that is the square-root function (or half density), the model is defined in the correspondingly constructed Riemannian manifold. To estimate the warping function, it is proposed to minimize the average Hellinger distance which also minimizes the average Fisher-Rao distance between densities. An optimization algorithm is introduced by estimating the smooth monotone transformation of the warping function. Asymptotic properties of the proposed estimator are discussed. Simulation studies demonstrate the performance of the proposed approach. Applying to a proteomic-imaging study from the Alzheimer's Disease Neuroimaging Initiative (ADNI), the proposed approach illustrates the connection between the distribution of protein abundance in the cerebrospinal fluid and the distribution of brain regional volume. Discrepancies among cognitive normal subjects, patients with mild cognitive impairment, and AD are identified and the findings are in line with existing knowledge about AD. 

As an initial step to study the association between two densities via the Riemannian representation, the current study focuses on the introduction of the model and related concepts, as well as the estimating procedure. For inference, the introduced strategy assumes that the true density functions are given. It only considers the uncertainty in estimating the warping function and ignores the uncertainty in estimating the density functions from the observed data, thus underestimating the variance with a narrower point-wise confidence interval. One future direction is to take the variation of density estimation into consideration. The current inference approach constructs a point-wise confidence interval for the estimated warping function. The construction of simultaneous confidence bands is also of future research.
The introduced regression model only considers the association between two densities. No other covariates are considered. Adding additional covariates, either scalar or functional or both, is not straightforward given the current formulation, thus is a direction of future research. In the ADNI application, the proposed approach is applied to the three diagnostic groups separately. Considering a linear regression model in the Euclidean space, a naive way of studying the discrepancy in the association across groups is to include an interaction term. Generalizing to the density-on-density regression model, it is to include an interaction term between the density predictor and the scalar group indicator. Such a generalization is not straightforward and requires further investigation.
In practice, both outcome and predictor densities may have subject-specific supports that could be informative for scientific questions of interest. Incorporating and accounting for the differences between subject-specific domains will need to be addressed in the future.
One other interesting direction of research is to study and model density-level residuals that may provide additional information for goodness-of-fit and identification of outliers. Comparing the warping function that links outcome and predictor densities with different parametric and non-parametric copula constructs may provide additional insights into co-dependence between underlying distributions.

\section*{Acknowledgments}

YZ and BC were partially supported by NIH grant R01MH126970. YZ was partially supported by NIH grants P30AG072976 and U54AG065181. BC was partially supported by NIH grants R01EB029977, P41EB031771, U54DA049110, and R01EB022911. AD was partially supported by NIH grant R01ES033739.

\appendix
\counterwithin{figure}{section}
\counterwithin{table}{section}
\counterwithin{equation}{section}
\counterwithin{lemma}{section}
\counterwithin{theorem}{section}
\counterwithin{proposition}{section}




\section{Theory and Proof}
\label{appendix:sec:proof}

\subsection{Additional geometries of \texorpdfstring{$\mathbb{S}_{\infty}$}{}}
\label{appendix:sub:geometries}

Let $p_{1}$ and $p_{2}$ be two points in $\mathbb{S}_{\infty}$ with $p_{1}\neq \pm p_{2}$. Then the parallel transport map from $T_{p_{1}}(\mathbb{S}_{\infty})$ to $T_{p_{2}}(\mathbb{S}_{\infty})$ along the shortest geodesic from $p_{1}$ to $p_{2}$ is
\[
  \Psi(v)=v-\frac{2vp_{2}}{|p_{1}+p_{2}|^{2}}(p_{1}+p_{2}), \quad \text{for } v\in T_{p_{1}}(\mathbb{S}_{\infty}).
\]
Model~\eqref{eq:dens_model} considers an error $e_{i}\in T_{p_{i}}(\mathbb{S}_{\infty})$ with respect to $p_{i}$. Assuming a base point in $\mathbb{S}_{\infty}$ denoted as $p_{0}$, there exists a transported error of $e_{i}$ in the tangent space of $T_{p_{0}}(\mathbb{S}_{\infty})$, which is denoted as $\varepsilon_{i}$. It is assumed that $\mathbb{E}\varepsilon_{i}=0$. Based on the transformation above, we have that $\mathbb{E}e_{i}=0$.


\subsection{Proof of Lemma~\ref{lemma:SRF_isometry}}
\label{appendix:sub:proof_lemma_SRF_isometry}

\begin{proof}
  \begin{eqnarray*}
    d^{2}((p_{1},\beta),(p_{2},\beta)) &=& \int \|\sqrt{f_{1}(\beta(\omega))\beta'(\omega)}-\sqrt{f_{2}(\beta(\omega))\beta'(\omega)}\|^{2}~\mathrm{d}\omega \\
    &=& \int \|\sqrt{f_{1}(\beta(\omega))}-\sqrt{f_{2}(\beta(\omega))}\|^{2}~\mathrm{d}\beta(\omega) \\
    &=& d^{2}(p_{1},p_{2}).
  \end{eqnarray*}
\end{proof}

For a density function $f\in\mathscr{P}$, let $[f]=\{f\odot\beta~|~\beta\in\Gamma_{\beta}\}$ be the equivalence class of $f$ induced by $\Gamma_{\beta}$. Under the $\mathbb{L}^{2}$-metric denoted as $d(\cdot,\cdot)$, for two densities $f_{1},f_{2}\in\mathscr{P}$, $[f_{1}]$ and $[f_{2}]$ are not parallel, that is $d(f_{1},f_{2})\neq d(f_{1}\odot\beta,f_{2}\odot\beta)$, and thus can be arbitrarily close to each other \cite[the so-called pinching effect,][]{marron2015functional}. Using the square-root representation (or the half density), the following isometry property is satisfied.
Let $[p]=\{(p,\beta)~|~\beta\in\Gamma_{\beta}\}$ denote the equivalent class of $p$ induced by $\Gamma_{\beta}$. Lemma~\ref{lemma:SRF_isometry} suggests that $[p_{1}]$ and $[p_{2}]$ are parallel under the $\mathbb{L}^{2}$-metric. The collection of such equivalence classes, $[p]$, is called a \textit{quotient space}, denoted by $\mathscr{Q}=\mathbb{L}^{2}([0,1])/\Gamma_{\beta}$. Based on this isometry property, the following defines the elastic distance between two density functions.
\begin{definition}[Elastic distance]
  For any two density functions, $f_{1},f_{2}\in\mathscr{P}$, and the corresponding SRFs, $p_{1},p_{2}\in\mathbb{L}^{2}$, define the elastic distance, denoted as $d$, on the quotient space $\mathscr{Q}$ to be
  \begin{equation}
    d([p_{1}],[p_{2}])=\inf_{\beta\in\Gamma_{\beta}}\|p_{1}-(p_{2},\beta)\|.
  \end{equation}
\end{definition}
\noindent The proposed estimator in~\eqref{eq:beta_est} thus aims to minimize an analogy of the elastic distance between two sets of densities in the sense of population average. Through the SRFs, it transforms the action on $\mathscr{P}$ to an action on $\mathbb{L}^{2}$.

\subsection{Additional properties of the Fisher-Rao metric}
\label{appendix:sub:FR-metric}

\begin{lemma}
  The Fisher-Rao metric is invariant under area-preserving mapping,
  \begin{equation}
    ((\tilde{v}_{1},\tilde{v}_{2}))_{\tilde{f}}=((v_{1},v_{2}))_{f},
  \end{equation}
  where $\tilde{f}=(f\circ\beta)\beta'$ and $\tilde{v}_{k}=(v_{k}\circ\beta)\beta'$ for $k=1,2$.
\end{lemma}

The following proposition demonstrates the infinitesimal equivalence of the Kullback-Leibler divergence to the Fisher-Rao metric.
\begin{proposition}[\cite{srivastava2016functional}]\label{prop:FR_KL}
  Let $f\in\mathscr{P}$ with $f>0$ and $v\in T_{f}(\mathscr{P})$ and denote the nonparametric Fisher-Rao metric as $((\cdot,\cdot))_{f}$. Then,
  \begin{equation}
    \lim_{\epsilon\rightarrow0}\left(\frac{\mathrm{KL}(f~\|~f+\epsilon v)}{((\epsilon v,\epsilon v))_{f}}\right)=\frac{1}{2},
  \end{equation}
  where $\mathrm{KL}(f_{1}~\|~f_{2})$ is the Kullback-Leibler (K-L) divergence of two density functions,
  \begin{equation}
    \mathrm{KL}(f_{1}~\|~f_{2})=\int f_{1}(\omega)\log\left(\frac{f_{1}(\omega)}{f_{2}(\omega)}\right)~\mathrm{d}\omega.
  \end{equation}
\end{proposition}

The following lemma transforms the Fisher-Rao metric in the space of probability distributions to the $\mathbb{L}^{2}$-metric through the SRF.
\begin{lemma}\label{lemma:FR-L2}
  Under the square-root mapping, the Fisher-Rao metric for probability densities transforms to the $\mathbb{L}^{2}$-metric up to a constant.
\end{lemma}
\begin{proof}
  For a density function $f\in\mathscr{P}$ and $f>0$, let $p(\omega)=\sqrt{f(\omega)}$. Denote $\mathscr{S}\subset\mathbb{S}_{\infty}^{+}$ as the set of all square-root forms or half-densities of probability density functions on $[0,1]$. $p\in\mathscr{S}$ and denote the corresponding tangent space as $T_{p}(\mathscr{S})$. Assume $v_{1},v_{2}\in T_{f}(\mathscr{P})$, then $v_{k}(\omega)=2\sqrt{f(\omega)}w_{k}(\omega)$, where $w_{k}\in T_{p}(\mathscr{S})$, for $k=1,2$. The the Fisher-Rao metric is
  \begin{eqnarray*}
    ((v_{1},v_{2}))_{f} &=& \int v_{1}(\omega)v_{2}(\omega)\frac{1}{f(\omega)}~\mathrm{d}\omega \\
    &=& 4 \int \sqrt{f(\omega)}w_{1}(\omega)\sqrt{f(\omega)}w_{2}(\omega)\frac{1}{f(\omega)}~\mathrm{d}\omega \\
    &\propto& \int w_{1}(\omega)w_{2}(\omega)~\mathrm{d}\omega \\
    &=& \langle w_{1},w_{2}\rangle_{\mathscr{S}},
  \end{eqnarray*}
  where $\langle\cdot,\cdot\rangle_{\mathscr{S}}$ is the $\mathbb{L}^{2}$ metric in $\mathscr{S}$.
\end{proof}

\subsection{Connections between the Hellinger distance and the Riemannian distance}

It is easy to show the equivalence of minimizing the Hellinger distance between two densities to minimizing the Riemannian distance. Assume $f_{1}(\omega)$ and $f_{2}(\omega)$ are two density functions in $\mathscr{P}$ and $p_{1}(\omega)=\sqrt{f_{1}(\omega)}$ and $p_{2}(\omega)=\sqrt{f_{2}(\omega)}$ are the corresponding square-root functions, respectively. The Riemannian distance (or geodesic length distance) between $f_{1}$ and $f_{2}$ is
  \begin{equation*}
    d_{\mathrm{R}}(f_{1},f_{2})=\cos^{-1}\left(\int_{0}^{1}\sqrt{f_{1}(\omega)}\sqrt{f_{2}(\omega)}~\mathrm{d}\omega\right)=\cos^{-1}\left(\int_{0}^{1}p_{1}(\omega)p_{2}(\omega)~\mathrm{d}\omega\right).
  \end{equation*}
  The Hellinger distance between $f_{1}(\omega)$ and $f_{2}(\omega)$ is
  \begin{equation*}
    H^{2}(f_{1},f_{2}) = 1-\int_{0}^{1}\sqrt{f_{1}(\omega)}\sqrt{f_{2}(\omega)}~\mathrm{d}\omega=1-\int_{0}^{1}p_{1}(\omega)p_{2}(\omega)~\mathrm{d}\omega.
  \end{equation*}
In the range of $[0,\pi]$, $\cos^{-1}$ is a monotonic decreasing function. Thus, the optimizer of minimizing $d(f_{1},f_{2})$ is the same as the one that minimizes $H(f_{1},f_{2})$ for $f_{2}=f_{1}\odot\beta$.

\subsubsection{Proof of Proposition~\ref{prop:beta_Riemannian}}
\begin{proof}
  Assume that $\beta'$ is also a warping function in $\Gamma_{\beta}$ different from $\hat{\beta}$, then
  \[
    \frac{1}{n}\sum_{i=1}^{n}H^{2}(g_{i},f_{i}\odot\hat{\beta})-\frac{1}{n}\sum_{i=1}^{n}H^{2}(g_{i},f_{i}\odot\beta')=\frac{1}{n}\sum_{i=1}^{n}\left\{H^{2}(g_{i},f_{i}\odot\hat{\beta})-H^{2}(g_{i},f_{i}\odot\beta')\right\}<0.
  \]
  Based on the relationship between the Hellinger distance and the Riemannian distance discussed above,
  \[
    d_{\mathrm{R}}(g_{i},f_{i}\odot\hat{\beta})=\cos^{-1}\left\{1-H^{2}(g_{i},f_{i}\odot\hat{\beta})\right\}\triangleq \cos^{-1}x_{i},
  \]
  \[
    d_{\mathrm{R}}(g_{i},f_{i}\odot\beta')=\cos^{-1}\left\{1-H^{2}(g_{i},f_{i}\odot\beta')\right\}\triangleq \cos^{-1}x_{i}',
  \]
  where $|x_{i}|,|x_{i}'|<1$. Using the Taylor expansion of $\cos^{-1}x$ for $|x|<1$,
  \begin{eqnarray*}
    \cos^{-1}x_{i}-\cos^{-1}x_{i}' &=& \left[\frac{\pi}{2}-\left\{x_{i}+\frac{1}{2}\frac{x_{i}^{3}}{3}+\mathcal{O}(x_{i}^{5})\right\}\right]-\left[\frac{\pi}{2}-\left\{x_{i}'+\frac{1}{2}\frac{x_{i}^{\prime 3}}{3}+\mathcal{O}(x_{i}^{\prime 5})\right\}\right] \\
    &=& (x_{i}'-x_{i})\left\{1+\frac{1}{6}(x_{i}^{\prime 2}+x_{i}'x_{i}+x_{i}^{2})\right\}+\mathcal{O}(x_{i}^{5})+\mathcal{O}(x_{i}^{\prime 5}). \\
  \end{eqnarray*}
  Then,
  \begin{eqnarray*}
    && \frac{1}{n}\sum_{i=1}^{n}d_{\mathrm{R}}(g_{i},f_{i}\odot\hat{\beta})-\frac{1}{n}\sum_{i=1}^{n}d_{\mathrm{R}}(g_{i},f_{i}\odot\beta') \\
    &=& \frac{1}{n}\sum_{i=1}^{n}\left[(x_{i}'-x_{i})\left\{1+\frac{1}{6}(x_{i}^{\prime 2}+x_{i}'x_{i}+x_{i}^{2})\right\}+\mathcal{O}(x_{i}^{5})+\mathcal{O}(x_{i}^{\prime 5}) \right] \\
    &\lesssim& \frac{3}{2}\cdot\frac{1}{n}\sum_{i=1}^{n}(x_{i}'-x_{i}) \\
    &=& \frac{3}{2}\cdot\frac{1}{n}\sum_{i=1}^{n}\left\{H^{2}(g_{i},f_{i}\odot\hat{\beta})-H^{2}(g_{i},f_{i}\odot\beta')\right\} \\
    &<& 0.
  \end{eqnarray*}
  Thus, $\hat{\beta}$ also minimizes the average Riemannian distance between the outcome and warped densities.
\end{proof}

\subsection{Proof of Theorem~\ref{thm:beta_asmp}}
\label{appendix:sub:proof_thm_beta_asmp}

\begin{proof}
  The space of the half densities, $\mathbb{S}_{\infty}$, is equipped with the $\mathbb{L}^{2}$-metric. As demonstrated in Section~\ref{sec:model}, after parallel transport, the model errors satisfy that $\mathbb{E}\varepsilon_{i}=0$. The proposed estimator is the minimizer of the average $\mathbb{L}^{2}$-distance in $\mathbb{S}_{\infty}$. From the law of large numbers, the consistency of the proposed estimator holds.
\end{proof}
The consistency in Theorem~\ref{thm:beta_asmp} can be considered as an analogy of the consistency of the ordinary least squares estimator in linear regression.
Proposition~\ref{prop:beta_Riemannian} shows the sufficiency that the proposed estimator also minimizes the average Riemannian distance between the outcome and warped density functions. From this perspective, it also demonstrates the consistency of the proposed estimator.

\subsection{A discussion of Theorem~\ref{thm:w_asmp}}
\label{appendix:sub:proof_thm_w_asmp}

Here, we first list the assumptions for the theoretical results. These assumptions may not be the weakest possible conditions. Improvements and relaxations of these assumptions will be a future direction.

\begin{description}
  \item[(A1)] $w(\omega)$ is in the reproducing kernel Hilbert space and uniformly continuous on $[0,1]$.
  \item[(A2)] $w(\omega)$ has finite $\mathbb{L}^{2}$-norm, i.e., $\|w(\omega)\|^{2}=\int_{0}^{1} w^{2}(\omega)~\mathrm{d}\omega<\infty$.
  \item[(A3)] Assumptions in Theorem~\ref{thm:beta_asmp} hold.
\end{description}

Assumptions (A1) and (A2) regulate the weight function in smooth monotone transformations. In the reproducing kernel Hilbert space, \citet{cox1983asymptotics} studied the asymptotic properties of nonparametric regression estimates obtained from smoothing splines. Together with Assumption (A3), the conclusions can be extended to the proposed estimator of $w(\omega)$. Under these regularity conditions, the proposed estimator can be viewed as a pointwise $M$-estimator. Thus, the asymptotic distribution of an $M$-estimator is derived as in the theorem.


\subsection{The Riemannian structure of \texorpdfstring{$\Gamma_{\beta}$}{}}
\label{appendix:sub:Rstructure_Gamma}

This section discusses the Riemannian structure of $\Gamma_{\beta}$, which is the space of the warping functions on $[0,1]$. $\Gamma_{\beta}$ is not a linear space. Analogous to the space of probability distributions, a proper mapping is necessary to transform the metric to the standard $\mathbb{L}^{2}$ metric. For $\beta\in\Gamma_{\beta}$, define its \textit{square-root slope function (SRSF)} as
\begin{equation}
  S(\beta):\Gamma_{\beta}\rightarrow\mathbb{R}, \quad \psi(\omega)\equiv S(\beta)(\omega)=\mathrm{sgn}(\beta'(\omega))\sqrt{|\beta'(\omega)|}=\sqrt{\beta'(\omega)},
\end{equation}
where $\mathrm{sgn}(\cdot)$ is the sign function and $\mathrm{sgn}(\beta'(\omega))=1$ as $\beta$ is smooth and strictly increasing. The space of $S(\beta)$, denoted as $\Psi=\{S(\beta)~|~\beta\in\Gamma_{\beta}\}$, is called the space of square-root densities (SRDs). Since $\beta(0)=0$, $S(\beta)$ is bijection, that is, for a given $\psi\in\Psi$,
\[
  S^{-1}(\psi)(\omega)=\int_{0}^{\omega}\psi^{2}(u)~\mathrm{d}u.
\]
$\psi$ is positive and has unit $\mathbb{L}^{2}$-norm,
\[
  \|\psi\|^{2}=\int_{0}^{1}\psi^{2}(\omega)~\mathrm{d}\omega=\int_{0}^{1}\beta'(\omega)~\mathrm{d}\omega=\beta(1)-\beta(0)=1.
\]
Thus, $\Psi=\{\psi:[0,1]\rightarrow\mathbb{R}^{+}~|~\|\psi\|^{2}=1\}$ is the \textit{positive orthant} of of the unit sphere in the Hilbert space $\mathbb{L}^{2}([0,1])$. The arclenth distance in $\Psi$ is then equivalent to the Fisher-Rao metric in $\Gamma$,
\[
  d_{\mathrm{FR}}(\beta_{1},\beta_{2})=d(\psi_{1},\psi_{2})=\cos^{-1}\left(\langle\psi_{1},\psi_{2}\rangle\right)=\cos^{-1}\left(\int_{0}^{1}\psi_{1}(\omega)\psi_{2}(\omega)~\mathrm{d}\omega\right).
\]

\subsection{Under the Wasserstein metric}
\label{appendix:sub:Wasserstein}

In this section, we demonstrate that under the Wasserstein metric, the warping function, $\beta$, is not  isometry. For two non-negative probability density functions $f_{1},f_{2}\in\mathscr{P}$, the Wasserstein distance between the two is defined as
\begin{equation}
  d_{\text{W}}(f_{1},f_{2})=\left\{\int_{0}^{1}\left(F_{1}^{-1}(p)-F_{2}^{-1}(p)\right)^{2}~\mathrm{d}p\right\}^{1/2},
\end{equation}
where $F_{i}$ is the cumulative distribution function of $f_{i}$, for $i=1,2$. Let $g_{i}=f_{i}\odot\beta$ and $G_{i}$ is the cumulative distribution function of $g_{i}$,
\[
  G_{i}(u)=\int_{0}^{u}g_{i}(\omega)~\mathrm{d}\omega=\int_{0}^{u}f_{i}(\beta(\omega))\beta'(\omega)~\mathrm{d}\omega=\int_{0}^{\beta(u)}f_{i}(\nu)~\mathrm{d}\nu
\]
\[
  F_{i}(u)=\int_{0}^{u}f_{i}(\omega)~\mathrm{d}\omega \quad \Rightarrow \quad G_{i}^{-1}(v)=\beta^{-1}(F_{i}^{-1}(v))
\]
\[
  d_{\text{W}}(g_{1},g_{2}) = \left\{\int_{0}^{1}\left(G_{1}^{-1}(p)-G_{2}^{-1}(p)\right)^{2}~\mathrm{d}p\right\}^{1/2} = \left\{\int_{0}^{1}\left(\beta^{-1}(F_{1}^{-1}(p))-\beta^{-1}(F_{2}^{-1}(p))\right)^{2}~\mathrm{d}p\right\}^{1/2}
\]
$d_{\text{W}}(g_{1},g_{2})=d_{\text{W}}(f_{1},f_{2})$ holds if and only if $\beta$ is the identity function. Thus, $\beta$ is not isometry under the Wasserstein metric.


\bibliographystyle{apalike}
\bibliography{Bibliography}

\end{document}